\documentclass[preprint,noshowframe]{imsart}
\usepackage[utf8]{inputenc}
\usepackage[T1]{fontenc}
\usepackage{dsfont}
\usepackage{amsmath}
\usepackage{bbold}
\usepackage{amsfonts}
\usepackage{amssymb}
\usepackage{amsthm}
\usepackage{siunitx}
\usepackage{stmaryrd}

\usepackage[numbers]{natbib}
\usepackage{csquotes}
\usepackage[xindy]{glossaries}
\usepackage{algorithm}
\usepackage{algpseudocode}
\usepackage{tikz}
\usepackage{forest} 
\usepackage{xcolor}
\usepackage{array}
\usepackage{booktabs}
\usepackage{ltablex}
\usepackage{rotfloat}
\usepackage{multirow, makecell}
\usepackage{arydshln}

\usepackage{lscape}

\usepackage{graphicx}

\usepackage{centernot}
\usepackage{babel}
\usepackage{subcaption}
\usepackage{adjustbox}
\usepackage{listings}
\usepackage[colorlinks,citecolor=blue,urlcolor=blue]{hyperref}

\usepackage{enumitem}

\theoremstyle{definition}
\newtheorem{definition}{Definition}

\theoremstyle{plain}
\newtheorem{theorem}{Theorem}
\newtheorem{proposition}{Proposition}

\newtheorem{corollary}{Corollary}

\theoremstyle{remarks}

\newcommand{\PP}{\mathbb{P}}

\newcommand{\Hzero}{\mathcal{H}_0}

\newcommand{\winsize}{m}

\newcommand{\V}{V}
\newcommand{\W}{W}

\startlocaldefs
\numberwithin{equation}{section}
\theoremstyle{plain}
\endlocaldefs

\graphicspath{{Images/}}

\begin{document}

		\begin{frontmatter}
		\title{Breakpoint based online anomaly detection}
		\runtitle{Breakpoint based online anomaly detection}
		\runauthor{E. Krönert et al.}
			\author[A]{\fnms{Etienne}~\snm{Krönert}\ead[label=e1]{etienne.kronert@worldline.com}},
        	\author[A]{\fnms{Dalila}~\snm{Hattab}\ead[label=e3]{dalila.hattab@worldline.com}}
	    	\and
			\author[B]{\fnms{Alain}~\snm{Celisse}\ead[label=e2]{alain.celisse@univ-paris1.fr}}
		
		\address[A]{FS Lab, Financial Services, Worldline, France\printead[presep={,\ }]{e1,e3}}
		\address[B]{SAMM, Paris 1 Panthéon-Sorbonne University, France\printead[presep={,\ }]{e2}}

		\begin{abstract}
			This paper proposes a new online anomaly detector for time series. Classically, anomaly detectors do not adapt well in real time to changes in the reference distribution. The novelty of our approach is to use breakpoint detection to adapt online to the new reference behavior of the time series. The statistical performance of the detector is theoretically ensured by a control on the FDR. The anomaly detector is empirically evaluated in depth to assess its capabilities and limitations.
		\end{abstract}
		\begin{keyword}[class=MSC]  
			\kwd[62L10]{}
		\end{keyword}
		\begin{keyword}
			\kwd[Anomaly detection, ]{}
			\kwd[Time series, ]{}
			\kwd[Breakpoint detection, ]{}
			\kwd[FDR]{}
		\end{keyword}
	\end{frontmatter}

    \bigskip\bigskip
    	\emergencystretch 3em
	\section{Introduction}\label{sec:intro}
Anomalies refer to observations that appear so atypical when compared to the others that it suggests they originate from a different process \cite{hawkins1980}. Anomaly detection has many applications, such as security \cite{PangANDEAAnomalyNovelty2022} or health \cite{LiSurveyHeartAnomaly2020}.
Machine learning is widely used in anomaly detection \cite{chandola2009anomaly}.
A model \cite{goldstein2016comparative} learns unsupervised the reference behavior from historical data. Then, the learned reference behavior is compared to the observed data to raise an alarm if the difference is too large.  
The variety of anomaly detectors presented in the reviews \cite{PangANDEAAnomalyNovelty2022, blazquez2021review} is useful to adapt to different patterns in time series such as trend, seasonality, and autocorrelation. 

The limitation of this approach is that the reference model is learned only once on the historical dataset, which assumes that the reference of the time series is the same over time.
However, there are data drifts where the reference behavior of the time series changes. If the model is not updated, the data points observed after the drift are detected as false positives.
To overcome this problem, most popular strategies consist of periodically retraining the model on a fixed-length window of data. Others use a sliding window of fixed length to continuously learn the reference. For example, Random Cut Forest \cite{guha2016robust} is a method inspired by Isolation Forest and adapted to real time. DiLOF \cite{na2018dilof} adapts LOF for real time. Periodic retraining and fixed length sliding windows do not account for the true dynamics of the time series.  

Second, a poorly calibrated anomaly detector can lead to alarm fatigue. An overwhelming number of alarms desensitizes the people tasked with responding to them, leading to missed or ignored alarms or delayed responses \cite{cvach2012monitor, blum2010alarms}. One of the reasons for alarm fatigue is the high number of false positives that take time to resolve \cite{solet2012managing, lewandowska2023determining}. To reduce alarm fatigue, it is necessary to theoretically ensure that the number of false positives is low.
In the article \cite{MarySemisupervisedmultipletesting2022, marandon2022}, the FDR in offline anomaly detection is controlled by the Benjamini-Hochberg (BH) procedure \cite{benjamini1995,Benjaminicontrolfalsediscovery2001}. However, this does not apply to the online case. In our previous article \cite{kronert2023fdr} we obtained control of the FDR using a multiple testing procedure.This previous article was limited to stationary time series, while this article assumes that the distribution of the time series can shift. The previous article provides a good introduction to the tools used in this article.

\begin{figure}[!b]
	\centering
	\includegraphics[width=0.6\linewidth]{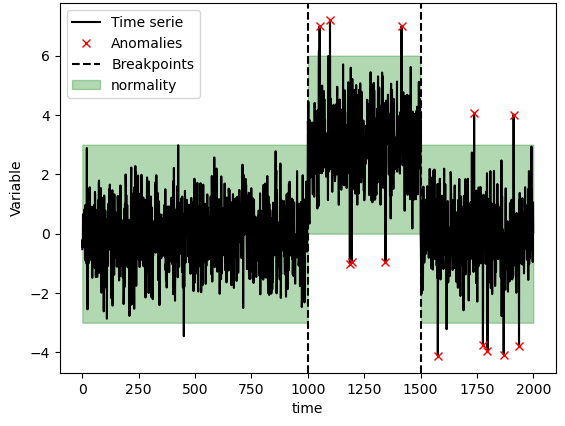}
	\caption{Anomaly detection based on breakpoints.}
	\label{fig:context-our-bkpts}
\end{figure}

This paper introduces a new anomaly detector that can update the learned reference behavior in an online context. As shown in Figure~\ref{fig:context-our-bkpts}, the main idea is to use a breakpoint detector to detect changes in the reference behavior of the time series. Breakpoints are the points at which a property of the time series changes. Between two breakpoints, the data form a homogeneous segment whose characteristics are easy to learn. After detecting the breakpoints in the time series, an atypicity score can be constructed by measuring the conformity of each point to its segment. The final step is to classify as anomalies the points with an atypicity score that is too high.  Note that unlike the proposals cited in \cite{aleneziDenialServiceDetection2013, fischRealTimeAnomaly2020}, the breakpoints do not correspond to anomalies, but to changes in the reference distribution. The use of a breakpoint detector introduces new difficulties, which are addressed in this paper. First, the detection of a breakpoint may be delayed, leading to temporary errors in segment assignment. Second, when a segment contains few points, it is difficult to estimate its behavior, generating anomaly detection errors. In an online context, this is particularly the case when points are observed just after a new breakpoint. This paper responds to these difficulties by assigning a confidence score to the estimation made by the detector. This score is used to judiciously select the estimates to be updated when their assigned confidence is too low. This confidence score is learned from a historical data set.

The anomaly detector presented in this paper comes with theoretical guarantees.
In a recent work of ours \cite{kronert2023fdr}, a new strategy has been designed in the online context to control the FDR for stationary series using a modified version of Benjamini-Hochberg applied to subseries. In the present paper, this work is extended to the nonstationary case.

{\bf The main contributions of this paper are summarized as follows:}
\begin{itemize}
	\item A versatile online anomaly detector based on breakpoint detection is built to adapt to changes in the reference behavior of the time series. Each component of the detector is studied in depth to provide the best possible parameters and improve the performance of the anomaly detector.
	\item The detector is theoretically studied to demonstrate its ability to control the FDR of the entire series at a level $\alpha$, under ideal hypotheses.
	\item The notions of active set and calibration set are introduced to deal with the difficulties of the online nature of the anomaly detector.
	\item The anomaly detector is empirically evaluated in numerous scenarios to determine its capabilities and limitations.
\end{itemize}

In Section~\ref{sec:method-description}, the problem of anomaly detection on piecewise iid time series is introduced, and the anomaly detector is described. In Section~\ref{sec:th-results} the main theorems are presented. The ability of the detector to control the FDR is empirically assessed in Section~\ref{sec:emp-reslt}. Finally, the detector is compared against competitors in Section~\ref{sec:evaluation-competitors}.

\section{Anomaly detection based on breakpoint detection}\label{sec:method-description}
This section introduces the new anomaly detector. First, the problem of anomaly detection in time series containing breakpoints is introduced in Section~\ref{sec:modelisation-problem}. Then, Section~\ref{sec:high-level} gives a high-level description of the detector. Finally, some design choices regarding the detector are discussed in Sections~\ref{sec:deal-unstationarity} and~\ref{sec:intro-uncertainty}.
\subsection{Modeling of the problem}
\label{sec:modelisation-problem}
Let $(\Omega, \mathcal F, \mathbb{P})$ be a probability space, with $\Omega$ the set of all possible outcomes, $\mathcal F$ a $\sigma$-algebra on $\Omega$ and $\mathbb P$ a probability measure on $\mathcal F$. Assume a realization of the independent random variables $(X_t)_{1\leq t\leq T}$, with $X_t$ taking values in a set $\mathcal X$ for all $t$. $T\in \mathbb{N}\cup\lbrace \infty \rbrace$ is the length of the time series. 
Normality is a concept that is dependent on a context that changes over time. 
The instants at which the reference distribution changes are called breakpoints. Supposing there are $D$ breakpoints where $D\in \mathbb N \cup \lbrace\infty \rbrace$, the position of the breakpoints is noted $(\tau_i)_1^D\in [1,T]^D$.
Unlike the modeling in \cite{kronert2023fdr}, which introduced a single reference distribution for the entire series, to model these different reference behaviors, several reference probability distributions are introduced and noted $\mathcal P_{0,i}$.
For each segment $i$ in $\llbracket 1, D\rrbracket$, for each point $t$ in this segment $\llbracket \tau_i, \tau_{i+1}-1\rrbracket$, the observation $X_t$ is called ``normal'' if $X_t \sim \mathcal P_{0,i}$. Otherwise $X_t$ is an ``anomaly''. 
Between two consecutive breakpoints, all ``normal'' observations are generated by the same law defining a homogeneous segment. The time series $(X_t)$ is piecewise stationary (except on anomalies).
As illustrated in Figure~\ref{fig:bkptad}, an observation $X_t$ is an anomaly if it is not generated from the reference distribution corresponding to the current segment. Figure~\ref{fig:bkptad} shows two anomalies detected in the second segment between breakpoints $\tau_2$ and $\tau_3$. Four anomalies have been detected in the last segment 3.

\begin{figure}[tb]
	\centering
	\includegraphics[width=0.6\linewidth]{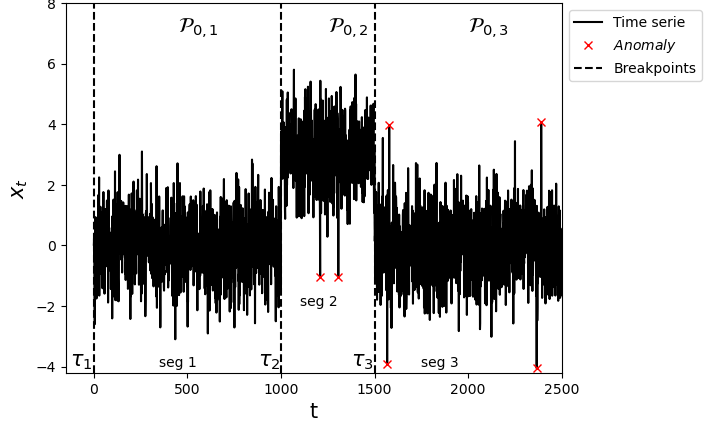}
	\caption{Illustration of piecewise stationary time series.}
	\label{fig:bkptad}
\end{figure}

The aim of an online anomaly detector is to find all anomalies among the new observations along the time series $(X_t)_{t\geq 1}$: for each instant $t>1$, a decision is taken about the status of $X_t$ based on past observations: $(X_s)_{1\leq s \leq t}$. 
Let $\mathcal H^{0}$ be the set of all normal data of the time series, $\mathcal H^{1}$ be the set of all abnormal data and $\mathcal R$ be the set of all data points detected by the anomaly detector. The FDR (resp. FNR) introduced in Section~\ref{sec:intro} can be expressed as the expectation of the False Discovery Proportion (FDP) (resp. False Negative Proportion (FNP)):
\begin{align*}
	FDR_1^T &= \mathbb E[FDP_1^T] = \mathbb{E}\left[ \frac{|\mathcal H^{0} \cap \mathcal R|}{|\mathcal R|\vee 1}\right]\\
	FNR_1^T &= \mathbb E[FNP_1^T] =  \mathbb{E}\left[ \frac{|\mathcal H^{1} \cap \mathcal R|}{|\mathcal H^{1}|\vee 1}\right]
\end{align*}
The control of the FDR at a targeted level $\alpha$ can be expressed by $FDR_1^T \leq \alpha$. In the following, the construction of an anomaly detector that controls the FDR at a desired level while minimizing the FNR is studied, in the case of piecewise stationary time series.

\subsection{High level description of the method}
\label{sec:high-level}
Using the statistical framework introduced in Section~\ref{sec:modelisation-problem}, the BreaKpoint detection based Anomaly Detector (BKAD) is introduced in Algorithm~\ref{alg:generalise-sbad} through the following steps.
\begin{enumerate}
	\item $\textbf{Breakpoint detection}$: A breakpoint detector estimates the number, $\hat D_t$, and the locations, noted $\hat \tau(t)_1, \ldots, \hat \tau(t)_{\hat D_t}$,  of breakpoints in the current time series $X_1^t = (X_1, \ldots, X_t)$. Consequently, the segments formed by two consecutive breakpoints are expected to be homogeneous. In particular, the segment formed between the last breakpoint noted $\hat b_t$ and the last observed point $t$ is called the current segment. With each new observation, the position of all the breakpoints is estimated again. In this way, a breakpoint estimated at one instant $t$ may disappear the next instant.
	Thanks to dynamic programming, the computational cost of estimating all breakpoints is limited.  For more precision, see Appendix~\ref{sec:breakpoint}.
	\item $\textbf{Active set selection}$: At this stage, the points whose status is to be reevaluated are selected. This set of points is called the active set. In the current segment, the points whose confidence in the previously evaluated status is too low are selected. The status of the other points remains the same as in the previous step. Two types of uncertainty are considered.  First,  uncertainty about the value of the atypicity score on short-length segments, if the current segment is shorter than the minimal requirement $\ell_\eta$, the active set contains the entire current segment. Second, uncertainty about the location of the breakpoints for observations that are too recent. Otherwise it contains only the last $\lambda_\eta$ data points whose segment assignment is uncertain. The values of $\ell_\eta$ and $\lambda_\eta$ are derived by $\hat f_d$ and $\hat f_\tau$. Methods for estimating $\hat f_\tau$ and $\hat f_d$ are described in Appendix~\ref{sec:delay-proba} and Appendix~\ref{sec:stable-seg-change}.
	\item $\textbf{Calibration set selection}$: The calibration set is used to calculate the $p$-values. Therefore, the calibration set should contain points that are representative of the reference behavior. Ideally, only points from the current segment should be used. But when the current segment doesn't contain enough points, points from other segments are used. To limit the bias caused by the introduction of points from another distribution, segments most similar to the current one are selected. The similarity between segments is measured using the similarity function $sim$. See Appendix~\ref{sec:calibration} for more details.
	\item $\textbf{Atypicity Score}$: As described in Appendix~\ref{sec:score}, a score $a:\mathcal X \rightarrow \mathbb R$ is a function reflecting the atypicity of an observation $X_t$, it aims to give a high value to anomalies. It is defined as a non conformity measure to the segment. The Nonconformity Measure $\overline a$, is a real valued function $\overline a(z, B)$ that measures how different $z$ is from the set $B$. A nonconformity measure can be used to compare a data point to the rest of the segment.
	\begin{align*}
		s_{u,t} = a(X_u) = \overline a(X_u, Seg_t(u)) \in \mathbb{R} 
	\end{align*}
	where $Seg_t(u)$ is the unique homogeneous segment that contains $X_u$, at time $t$. The NCM must be carefully chosen to be robust to the presence of anomalies in the current segment and to distinguish anomalies even with few points in this segment. 
	\item $\textbf{$p$-value estimator}$: The value of the atypicity score cannot be interpreted directly. The atypicity score assigned to a data point is compared with those assigned to the points in the calibration set. The probability of observing a normal data point with an atypicity score $a(X)$ greater than $a(X_t)$ is estimated. This is done using the empirical $p$-value estimator and the calibration set. See Appendix~\ref{sec:threshold} for more details. 
	\begin{align*}
		\hat p_e(s_{u,t}, \mathcal S^{cal}_t) = \frac{1}{|\mathcal S^{cal}_t|}\sum_{s \in \mathcal S^{cal}_t}\mathbb{1}[s > s_{u,t}]
	\end{align*}
	\item $\textbf{Threshold Choice}$: In order to control the FDR of the complete time series, the data-driven threshold is calculated from the empirical p-values of the active set. A multiple testing procedure, inspired from Benjamini-Hochberg, is applied to determine this detection threshold.  See Appendix~\ref{sec:threshold} for more details. This procedure was introduced in \cite{kronert2023fdr}. Abnormal status ($d_{u,t}=1$) is assigned to data points with a $p$-value below the threshold. 
	
\end{enumerate}

\begin{algorithm}[p]
	\begin{algorithmic}[1]
		\Require Let $T>0$ be the time series length, $(X_t)_1^T$ be the time series, $breakpointDetection$ implements breakpoint detector, $\hat f_\tau$ estimate the probability of segment assignment change and $\hat f_d$ are estimate the probability of status change when the breakpoint do not change, $sim$ a similarity function between segments, $\overline a$ is a non conformity measure, $\hat p_e$ implements the empirical $p$-value estimator and $\hat \varepsilon$ selects the best threshold to be applied.
       \State $\hat\ell_\eta \gets \arg\min\left\lbrace \ell, \hat f_\tau(\ell) < \eta\right\rbrace$ 
        \State $\hat\lambda_\eta \gets \arg\min\left\lbrace \lambda, \hat f_d(\lambda) < \eta\right\rbrace$
		\For{$t=1$ to $T$}
		\State $\hat \tau(t) \gets \text{breakpointDetection}(X_1^t)$ \Comment{Detection of the breakpoints}
		\State $b_t \gets \hat \tau(t)_D$
		\If{$t-\hat b_t \leq \hat\ell_\eta$}  \Comment{Definition of the active set}
		\State $m_t = t-\hat b_t$
		\Else
		\State $m_t = min(t-\hat b_t, \hat\lambda_\eta)$
		\EndIf
		\State $\mathcal{I}^{active} = \lbrace X_{t-m_t}, X_{t-m_t+1},\ldots, X_t\rbrace$
		\For{$i=1$ ito $\hat D_t$}   \Comment{Definition of the calibration set}
		\For{$u=\hat\tau_{i}(t)$ to $\hat\tau_{i+1}(t)$}
		\State $sim_u \gets sim(X_{\tau_i(t)}^{\tau_{i+1}(t)-1}, X_{\hat b_t}^t)$
		\EndFor
		\EndFor
		\State $sortedU = sort(\llbracket 1 ,t\rrbracket, sim)$
		\State $filteredU = filter(u \in sortedU, d_{u,t-1}=0)$
		\State $\mathcal I^{cal} \gets \lbrace filteredU_i, i \in \llbracket 1, n \rrbracket\rbrace$ 
		\State $\mathcal S^{cal} \gets \lbrace \overline a(X_u, Seg(u)), u \in \mathcal I^{cal}\rbrace$
		\For{$u$ in $\mathcal I^{active}$} \Comment{Computation of the scores}
		\State $s_{u,t} \gets \overline a(X_u, Seg(u))$
		\EndFor
		\For{$u$ in $\mathcal I^{active}$}   \Comment{Estimation of the $p$-values}
		\State $\hat p_{u,t} =  \hat p_e(s_u, \mathcal{S}^{cal})$
		\EndFor
		\State $\hat \varepsilon_t = \hat \varepsilon(\lbrace \hat p_{u,t}, u \in \mathcal I^{active}\rbrace)$ \Comment{Threshold estimation enabling FDR control}
		\For{$u$ in $\mathcal I^{active}$}
		\If{$\hat p_{u,t} < \hat\varepsilon_t$}  \Comment{Computation of the status}
		\State $d_{u,t} = 1$
		\Else
		\State $d_{u,t} = 0$
		\EndIf
		\EndFor
		\For{$u$ in $[1,t]\backslash\mathcal I^{active}$}
		\If{$t-\hat b_t<m $ and $u \geq \hat b_t -m$}   \Comment{Segment closed}
		\State $d_{u,t} = d_{u,\hat b_t}$
		\Else
		\State $d_{u,t} = d_{u,t-1}$
		\EndIf
		\EndFor
		\EndFor
		\\{\bf Output:} $(d_{t,T})_{t=1}^T$ boolean list that represent the detected anomalies.
	\end{algorithmic}
	\caption{Breakpoints based anomaly detection}
	\label{alg:generalise-sbad}
\end{algorithm}

\newpage
\begin{figure}[tb]
	\centering
	\includegraphics[scale=0.5]{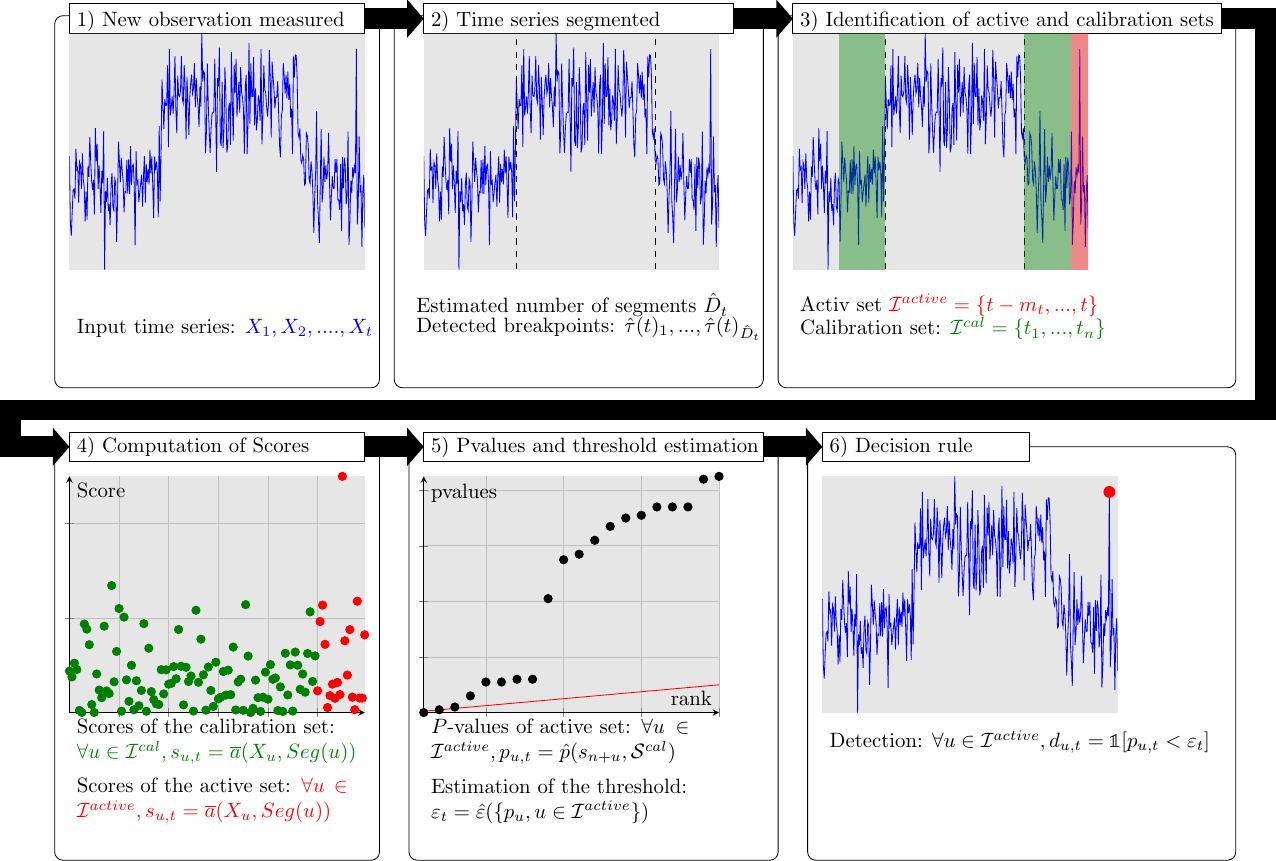}
	\caption{Description flow of Algorithm \ref{alg:generalise-sbad}. \label{tikz:win}}
\end{figure}

Algorithm~\ref{alg:generalise-sbad} is illustrated in Figure~\ref{tikz:win}, the description of the flow is given as the following:
\begin{enumerate}[label=Step~\arabic*,start=0]
	\item (not illustrated): the minimum number of points $\ell_\eta$ that a segment must contain to ensure that the atypicity score is estimated with sufficient accuracy is estimated. Similarly, the minimum delay $\lambda_\eta$ to ensure with high probability that the assignment of a point to a segment does not change is estimated.
	\item: for each time step $t$, a new data point $X_t$ is observed.
	\item: the current time series is segmented $\hat \tau(t)$. Each segment is homogeneous.
	\item: the data points having a status with low confidence are identified to build the active set. If the current segment is shorter than the minimal requirement $\ell_\eta$, the active set contains the entire current segment. Otherwise it contains only the last $\lambda_\eta$ data points whose segment assignment is uncertain. The calibration set is built by sorting the data points according to the similarity, then $n$ data points with the highest similarity are added to the calibration set.
	\item: The calibration set and active set data points are scored, using the non conformity measure $\overline a$.
	\item: The $p$-values of the active set are estimated using the calibration set.
	The multiple testing procedure is applied to the active set to obtain the data-driven threshold, in the figure the threshold is chosen using the Benjamini-Hochberg procedure.
	\item: A decision is made to give the abnormal status to the data point with a $p$-value lower than the threshold. For points outside the active set, their status remains the same as in the previous step. If a current segment has less than $m$ points, it is considered a new segment. In this case, the previous segment has just been closed by a new breakpoint. The status of the data point preceding the new breakpoint is updated using the most relevant historical status. This status is the last one before observing the data of the current segment and biasing the status estimation ($d_{u,t}=d_{u,\hat b_t}$).  
\end{enumerate}
The modularity of our method allows a better adaptation to the diversity of time series.
Two design choices are discussed: the use of training, calibration, and test sets in Section~\ref{sec:deal-unstationarity} and the need to re-evaluate estimations in Section~\ref{sec:intro-uncertainty}.
\subsection{Training, calibration and test set.}
\label{sec:deal-unstationarity}
Let's look back at some of the design choices made for the anomaly detector to understand the rationale behind them.
Emphasis is placed on the differences between this anomaly detector and one that analyzes stationary data.
First, let's look at the three time series subsets that play a role in anomaly detection. The following is a description of each:
\begin{itemize}
	\item  \textbf{Training set}: 
	This is the set of points used as a reference by the atypicity score. The atypicity score $a$ compares the observation $X_t$ to the \emph{training set} $\mathcal{X}^{train} = \lbrace X_1,....,X_q \rbrace$. The more $X_t$ deviates from the points in the training set, the more the abnormality score $s_t = a(X_t, \mathcal{X}^{train})$ is high.
	\item \textbf{Calibration set}:
	The calibration set is used by the $p$-value estimator to calibrate the score obtained.
	A $p$-value estimator $\hat p$, based on a \emph{calibration set} of scores $\mathcal S^{cal} = \lbrace s_{t-m-n},...,s_{t-m}\rbrace$ containing scores of data points generated from $\mathcal P_0$, to estimate the $p$-value, $\hat p_t = \hat p(s_t, \mathcal S^{cal})$. 
	\item \textbf{Test set}: The value of the threshold $\varepsilon$ can be chosen to be data driven to allow better control of the False Discovery Rate (FDR). The threshold is estimated using a subseries of $p$-values, $\lbrace p_{t-m+1},\ldots, p_{t}\rbrace$, called the test set. This denomination has been chosen to correspond to that of the offline case \cite{marandon2022}.
\end{itemize}

In case of stationary data, the training set and calibration set are either chosen from the start of the time series labeled with anomalies or evolve over time using sliding windows. When the training set cannot be labeled, a robust atypicity score is required. 
An example of a training set, calibration set and test set, in the context of online anomaly detection is shown in the following:
\begin{align*}
	\underbrace{X_1,\ldots,X_{q}}_{\mbox{Training set}},\ldots \underbrace{X_{t-n-m},\ldots,X_{t-n}}_{\mbox{Calibration set}},\underbrace{X_{t-m},\ldots,X_{t}}_{\mbox{Test set}}
\end{align*}
For each new observation $X_t$, the function $a$ is used to measure the atypicity of the point relative to the training set. The value of the score cannot be interpreted directly because the distribution of the scores under $\mathcal H_0$ is unknown. So its $p$-value is estimated using the calibration set. The more the data point is atypical, the closer the $p$-value is to $0$. The data-driven threshold is estimated from the test set.

Suppose this strategy used for stationary data is applied to a time series where a shift in the mean of the reference distribution occurs. Before the first shift, there are no differences with the stationary case. After the shift, all data points appear as anomalies when using the scoring function trained on the initial training set based on data before the shift. To adapt to the shift, the training and the calibration sets have to be rebuilt on the new segment of data in order to reapply the anomaly detector. 
\begin{align*}
	\underbrace{X_1,\ldots,X_{\tau_1}}_{\mbox{Segment 1}}, \underbrace{\underbrace{X_{\tau_1+1},\ldots,X_{\tau_1+q}}_{\mbox{train}},\underbrace{X_{\tau_1+q+1},\ldots,X_{\tau_1+q+n}}_{\mbox{calibration}},\underbrace{X_{\tau_1+q+n+1},\ldots, X_t}_{\mbox{test}}}_{\mbox{Segment 2}}
\end{align*}
However, it would take a lot of time to gather enough data for the training and calibration sets. This is the reason why two improvements are suggested.
The first improvement in the case where the score is stationary across different segments, data for the calibration set can be taken from previous segments. For example, suppose the shift occurs in the mean  and the score is the $z$-score: $(x-\mu)/\sigma$.
\begin{align*}
	\underbrace{\underbrace{X_{1},\ldots,X_{q}}_{\mbox{train}},\underbrace{X_{q+1},\ldots,X_{q+m}}_{\mbox{calibration}}, \ldots}_{\mbox{Segment 1}}, \underbrace{\underbrace{x_{\tau_1+1},\ldots,X_{\tau_1+q}}_{\mbox{train}},\underbrace{X_{\tau_1+q+1},\ldots,X_t}_{\mbox{test}}}_{\mbox{Segment 2}}
\end{align*}
Furthermore, if the scoring function is robust to the presence of anomalies inside the training set, the training can contain little amount of anomalies. The whole segment can be used as training set. The test set can be part of the training set, using a leave-one-out strategy. The segment length required for anomaly detection can thus be further reduced, this constitute the second improvement. These set constructions appear throughout Algorithm~\ref{alg:generalise-sbad}.
\begin{align*}
	\underbrace{\underbrace{X_{1},....,X_n}_{\mbox{calibration}}, \ldots}_{\mbox{Segment 1 and train}}, \underbrace{\ldots, \underbrace{X_{t-m}......X_t}_{\mbox{test}}}_{\mbox{Segment 2 and train}}
\end{align*}

\subsection{The need to re-evaluate estimates}
\label{sec:intro-uncertainty}
In Algorithm~\ref{alg:generalise-sbad}, unlike classic online anomaly detection, the status assigned to a data point can change over time.
Indeed, if all data points are generated by the same distribution, there's no need to perform error-prone segmentation on the time series. Furthermore, it can be guaranteed that once the series size is sufficient, there are enough points to accurately calculate the atypicity score and the $p$-value. As a result, there's no need to continuously evaluate the status of the points after each observation. The status of a point is estimated at the time of observation and does not change thereafter.

When the time series under study has shifts in its reference distribution, the construction of the training and calibration sets described in the previous section relies on the knowledge of the breakpoint locations. In practice, neither the number of segments $D$, nor the positions of the breakpoints $\tau_i$ nor the laws of the segments $\mathcal P_{0,i}$ are known. All these quantities must be learned using the breakpoint detector and the scoring function to perform anomaly detection.
Moreover, in an online context, the lack of  knowledge of the whole series influences a good estimation of these quantities and has a negative impact on the quality of the detection. With each new observation, different situations may occur: the position of a previous breakpoint may be adjusted or removed, or a new breakpoint may appear. As a result, the segment assigned to a data point changes. These new observations influence the composition of each segment and therefore modify the score value and the status assigned to each point, especially if the segment is small. Consequently, the values associated with a data point $X_u$ change over the time $t$. To reflect this evolution, a subscript $t$ is added. For example, $\hat p_{u,t}$ is the $p$-value estimated for $X_u$ at time $t$. Similarly $d_{u,t}$ is the status of the point $X_u$ at time $t$.
In addition, the concept of the active set is introduced to capture the most recently observed points in an online context, whose ``abnormal'' or ``normal'' status is uncertain as it may evolve due to the introduction of new data points. This uncertainty arises from the possibility that the segment to which a point is assigned may change over time, or from the estimation of scores on small segments.

Having described the anomaly detector in this section, the following Section~\ref{sec:th-fdr} examines the anomaly detector theoretically.
\section{Theoretical results}\label{sec:th-results}
After describing the anomaly detector in Section~\ref{sec:method-description}, the main theoretical results are given in this section. It is shown in Section~\ref{sec:th-fdr} that under ideal assumptions BKAD can control the FDR to a desired level $\alpha$.  Section~\ref{sec:th-uncertainty} shows that the proportion of errors committed by the detector due to the online context can be controlled to a desired level by correctly building the active set. Section~\ref{sec:hypo-disc} discusses the validity of the ideal hypotheses.
\subsection{Control of the FDR}\label{sec:th-fdr}
We have studied the control of the FDR of an online detector in the case of iid time series (without breakpoints) in a previous article \cite{kronert2023fdr}. In this section, the results are extended to the case where the time series has breakpoints.
The various assumptions involved in the control of the FDR are introduced, followed by the presentation of the theorem.

The first hypothesis concerns the generation of the true anomalies. To be able to control the FDR of the whole time series from a control on subseries, it is necessary that the proportion of anomalies in subseries is the same as to the rest of the series. The classical assumption is that the data points are generated by a mixture of a reference distribution and an alternative distribution \cite{huber1992robust, staerman2022functional} .
\begin{definition}\label{def:uniform-pi}[Time series with uniform proportion of anomalies]
	Let $D$ be the number of segments. Let $\tau_1,\ldots,\tau_D$ be the breakpoint locations. Let $\mathcal P_{0,1}, \ldots, \mathcal P_{0,D}$ be the reference distributions and $\mathcal P_{1,1}, \ldots, \mathcal P_{1,D}$ be the alternative distributions. Let $\pi$ be the proportion of anomalies. 
	A time series is said to have a uniform proportion of anomalies if $(A_u)$ the series describing anomaly locations and $(X_u)$ the series of observations are generated as follows:
	\begin{align}
		\forall i \in \llbracket 1, D\rrbracket, \forall u \in \llbracket \tau_{i}, \tau_{i+1}-1 \rrbracket,\qquad A_u \sim Ber(\pi) \mbox{ and } X_u \sim \begin{cases}
			\mathcal P_{0,i},& \mbox{if } A_u=0\\
			\mathcal P_{1,i},& \mbox{if } A_u=1
		\end{cases}
	\end{align}
\end{definition}

The second key assumption for FDR control is that breakpoints must be identified without error. However, in an online context, breakpoint detection is subject to some time delay. To account for this, it is assumed that there may be errors in the most recent observations, but that beyond $\lambda^*$ data points, all breakpoints are correctly detected.
\begin{definition}\label{def:ideal-bkpt}[Ideal breakpoint detector with delay $\lambda^*$]
	Let $\tau$ be the true segmentation.
	Let $\hat \tau$ be the breakpoint detector, where for all $t$, $\hat \tau(t)$ is the estimated segmentation of $X_1, \ldots, X_t$. 
	$\hat\tau$ is called an ideal breakpoint detector with delay $\lambda^*$ if the true segmentation is found with delay $\lambda^*$.
	\begin{align}\label{assume.segmentation}
	 \forall t \in \llbracket 1, T \rrbracket,\quad	\llbracket 1, t-\lambda^*\rrbracket \cap \hat\tau(t) = \llbracket 1, t-\lambda^*\rrbracket \cap \tau\tag{{\bf Segmentation}}
	\end{align}
\end{definition}

As discussed in Section~\ref{sec:deal-unstationarity}, it is desirable for the computed scores to follow the same distribution for all segments to correctly estimate the $p$-value.
For example, if each segment $i$ follows a reference distribution $\mathcal N(\mu_i,1)$, then only the mean changes at the breakpoints and the oracle score $\tilde a(X_u,i) = |X_u - \mu_i|$ is iid. In practice, however, the mean $\mu_i$ is not known, so the empirical mean of the segments is used. In doing so, the independence property between the scores is lost. 
But since $\hat\mu_i$ converges to $\mu_i$, it can be assumed that for a segment of sufficient length, the scores can be considered iid. This idea is formalized by the property ``Score idd for minimal segment length''.
\begin{definition}\label{def:score-homogeneity}[Score idd for minimal segment length]
	Let $(X_u)$ be a time series satisfying Definition~\ref{def:uniform-pi}. 
	Assumption \ref{assume.score} assumes that there exist an oracle score  $\tilde a$, such that $\tilde a(X_u, i_u) = s_u$ is iid, where $i_u$ is the number of the segment to which $u$ belongs. 
	Furthermore \ref{assume.score} assumes there is a non conformal measure $\overline a$ and an integer $\ell^*$ such that:
	\begin{align}\label{assume.score}
		\forall i \in \llbracket 1,D\rrbracket, &\forall u_1, u \in \llbracket\tau_i,\tau_{i+1}-1\rrbracket, \quad |\tau_i-u_1|\geq \ell^*,\notag\\
		& \quad \overline a(X_u, \lbrace X_{\tau_i},\ldots,X_{u_1}\rbrace)=\tilde a(X_u,i)\tag{{\bf Score}}
	\end{align}
\end{definition}
This property is verified if, on the one hand, there is an oracle score whose distribution is iid. And secondly, the non-conformity measure must converge towards this oracle score.

Finally, an ideal version of BKAD is introduced to facilitate theoretical study. 
The ideal BKAD algorithm is described in the following Definition~\ref{def:ideal-bkad}. This is an ideal version of the algorithm presented in Algorithm~\ref{alg:generalise-sbad}, assuming no computational constraints and that the true labels are known when building the calibration set.
At each time step $t$, the scores and $p$-values are updated with information from the new observed data point, then the $d_{u,t}$ status is changed in two cases: 
\begin{itemize}
	\item for the most recent observations,
	\item When a new segment is detected, the status of the last points of the closed segment is updated.
\end{itemize}
In other cases, $d_{u,t}$ keeps its value computed at the previous instant.

\begin{definition}\label{def:ideal-bkad}[Ideal BKAD]
	Let $\lambda'$ and $\ell'$ be two parameters.
	Noting $m=max(\lambda', \ell')$, for each $t$ in $\llbracket 1, T\rrbracket$, the series of scores $(s_{u,t})$, $p$-values $(p_{u,t})$ and decision $(d_{u,t})$ of the ideal BKAD are calculated as the following.
	
	First, the sequence of scores is computed as follows. The calculation is presented separately for the segments identified between two breakpoints and for the current segment:
	\begin{align}
			\forall i \in \llbracket 1, \hat D_t-1\rrbracket, &\forall u \in \llbracket \hat \tau_{i}(t) , \hat \tau_{i+1}(t)-1 \rrbracket,\notag\\
			\quad s_{u,t} &= \overline a(X_u, \lbrace X_{\hat\tau_{i}(t)},\ldots,X_{min(\hat\tau_{i+1}(t)-1,t-m)} \rbrace)\\
      \forall u \in \llbracket \hat b_t , t \rrbracket,\quad s_{u,t} &= \overline a(X_u, \lbrace X_{\hat b_t},\ldots,X_{t-m} \rbrace)\label{eq:score-ideal-bkad}
	\end{align}
	
	Then, the sequence of $p$-values is computed as follows: $\hat p_{u,t} = \hat p_e(s_{u,t}, \mathcal S_t)$.
	With $\mathcal S_t$ be the calibration at time step $t$, it is computed as follows:
	\begin{align}
		\mathcal S_t = \lbrace s_{h(t-m,1),t},\ldots, s_{h(t-m,n),t}\rbrace\label{eq:pvalue-ideal-bkad}
	\end{align}
	The calibration set is a sliding window containing the $n$ previous scores generated according to the reference distribution. For each $t$ and $i$, $h(t,i)$ gives the $i$-th observation lower than $t$ that satisfies the $\Hzero$ hypothesis.
	
	Finally,  $(d_{u,t})$ the series of decisions, is computed as follows:
	\begin{itemize}
		\item The status of the most recent observations is updated:
		\begin{align}
			\forall u \in \llbracket max(t-m,\hat b_t),t\rrbracket,\qquad d_{u,t}=\mathbb{1}[p_{u,t}<\hat \varepsilon(p_{t-m,t}, \ldots, p_{t,t})] \label{eq:end-seg-ideal-bkad}
		\end{align}
		\item If needed, the status of the last points of the previous segment is updated:
		\begin{align}
			\forall u \in \llbracket \hat b_t-m,\hat b_t-1\rrbracket,\qquad d_{u,t}=\mathbb{1}[p_{u,t}<\hat \varepsilon(p_{\hat b_t-m,t}, \ldots, p_{\hat b_t-1,t})]\label{eq:closing-seg-ideal-bkad}
		\end{align}
		\item The status of other data points remains unchanged.
		\begin{align}
			d_{u,t} = d_{u,t-1} \label{eq:continu-ideal-bkad}
		\end{align}
	\end{itemize}
\end{definition}
The detector associates each observed data point $X_u$ with a status $d_{u,t}$ that can evolve according to the number of observed points $t$. The goal of the ideal detector is that the $d_{u,t}$ value converges to a final $d_{u,T}$ value in a small number of steps and that the final decision series controls the FDR at a desired level $\alpha$ with a minimum of false negatives.
Under assumptions \ref{assume.segmentation} and \ref{assume.score}, the ideal version of BKAD, described in Definition~\ref{def:ideal-bkad}, controls the FDP of the complete series at the level of the mFDR of the subseries of length $m$.
\begin{theorem}[False Discovery Proportion convergence]\label{thm:ideal-detector}
	Let $(X)_{t\geq 1}$ be a time series of infinite size with uniform proportion of anomalies $\pi$ as stated in Definition~\ref{def:uniform-pi}.
	It is assumed that assumptions \ref{assume.segmentation} and \ref{assume.score} are verified. Applying the ideal BKAD with $\lambda'=\lambda^*$ and $\ell'=\ell^*$ on $(X_t)$, and noting $R_a^b$ the number of rejections on a subset $[a,b]$ and $FP_a^b$ the number of false positives on the same subset.
	\begin{align*}
		R_a^b &= \lim_{t \rightarrow \infty}\sum_{u=a}^b d_{u,t}\\
		FP_a^b &= \lim_{t \rightarrow \infty}\sum_{u=a}^b (1-A_u)d_{u,t}
	\end{align*}
	
	Then, the FDP, computed as $FDP_1^t = \frac{FP_{1}^t}{R_1^t}$, converges and its limit can be calculated as follows:
	\begin{align}
		\lim_{t \rightarrow \infty} FDP_1^t = mFDR_1^m
	\end{align}
\end{theorem}

The proof of this theorem is given in Appendix~\ref{appendix:proof-bkad-fdr}.
It follows from this theorem that to control the FDP at a desired $\alpha$ level, it is sufficient to control the mFDR on a subseries of length $m$ at the same level. Theorem~\ref{thm:ideal-detector} is an extension of Theorem 4 in \cite{kronert2023fdr} to time series containing breakpoints where the reference distribution of the time series changes.
 According to \cite{kronert2023fdr}, the modified BH procedure allows to control the mFDR if the $p$-values in the subseries are calculated with a unique calibration set, as stated in \cite{marandon2022}.

\begin{corollary}\label{corollary-fdr-bkad}
	Under the same notations and assumptions as Theorem~\ref{thm:ideal-detector}, let $m$ and $\nu$ be two integers, let $n$ and $\alpha'$ defined by:
	\begin{align*}
		\alpha' = \frac{\alpha}{1 + \frac{1-\alpha}{m\pi}} \mbox{ and }n = \nu m/\alpha'-1
	\end{align*}
	the threshold procedure is the Benjamini-Hochberg procedure with level $\alpha'$, also called the modified Benjamini-Hochberg procedure introduced in \cite{kronert2023fdr}.
	
	The number of data points detected as anomaly by BH in $\llbracket 1,m\rrbracket$ is noted $R_{1}^m$. Similarly, $R_{1}^m(u)$ represents the number of data points detected as anomaly, when $\hat p_{u,t}$ is replaced with 0. Assuming that the following assumption hold (for more precision refer to \cite{kronert2023fdr}):
	\begin{align}
		\mathbb E [R_{1}^m] \approx \frac{\winsize\pi}{1-\alpha}. \mbox{ and } \mathbb E[R_{1}^m(u)] = \mathbb E[ R_{1}^m]+1\label{eq:heuristic.power}
	\end{align}
	 
	 then the FDP of the complete time series is controlled almost surely at the level $\alpha$:
	\begin{align}
		\lim_{t \rightarrow \infty}FDP_1^t = (1-\pi)\alpha
	\end{align}
\end{corollary}

From Corollary~\ref{corollary-fdr-bkad}, the modified Benjamini-Hochberg procedure introduced in \cite{kronert2023fdr} allows to control the FDP at the desired level $\alpha$. In anomaly detection $1-\pi$ is close to $1$ since the number of anomalies is small. 
The almost surely convergence of the FDP implies the control of the FDR at level $\alpha$.
To maximize the performance of the anomaly detector, it is important to carefully choose the cardinality of the calibration set. Indeed, $n$ must be of the form $\nu m/\alpha'-1$ to ensure FDR control. See Section 3.2.2 in \cite{kronert2023fdr} for more details. The paper \cite{kronert2023fdr} conducts experiments to test the validity of the assumptions of Eq.~\ref{eq:heuristic.power}.

FDR control by the anomaly detector is an important property. This result guides the choice of the threshold selection procedure and the tuning of the calibration set cardinality in Algorithm~\ref{alg:generalise-sbad}, see more Appendix~\ref{sec:threshold} for more details.
However, \ref{assume.segmentation} and \ref{assume.score} are strong assumptions, it is not possible to get perfect estimations, the following Section~\ref{sec:th-uncertainty} studies the uncertainty of the estimations.
\subsection{Manage uncertainty of estimations}\label{sec:th-uncertainty}
In Section~\ref{sec:intro-uncertainty}, uncertainty about the location of breakpoints and the value of scores leads to the need to re-evaluate the status of data points. In the previous section, it was specified that under conditions \ref{assume.segmentation} and \ref{assume.score} then the status of points need to be re-evaluated only if the current segment is shorter than $\ell^*$ or the point was observed less than $\lambda^*$ steps ago. Otherwise, since the true breakpoint locations and segment values were known, there was no reason to re-evaluate the data point status. In practice, however, such conditions cannot be verified, and there are always errors in the estimates.  At each step, there is a set of data points whose status must be re-evaluated, called the active set. This raises the question of how to define the active set in a way that minimizes the estimation error of the status while limiting the number of re-evaluations.

Assuming that it is possible to estimate the minimum segment length $\hat\ell_\eta$, which guarantees that the score value is estimated with ``good accuracy'', and the delay $\hat\lambda_\eta$, which guarantees that the data point is ``correctly assigned to a segment'', then the ``correct way'' to construct the active set of data points to re-evaluate is suggested as follows.
As shown in Algorithm~\ref{alg:active-set}, the procedure starts by comparing the length $\ell_t$ of the current segment with the threshold length $\hat \ell_\eta$.
If the length $\ell_t$ is lower than this threshold, the whole segment is considered as the active set since the segment does not contain enough points to estimate the atypicity score with high precision. Otherwise, the segment contains enough points and the source of the status change is segment reassignment.
Considering the data points whose distance to the end of the time series is less than $\hat \lambda_\eta$, the risk of being reassigned to another segment is high.
Consequently, the active set will contain all points that are after  the position $t - \hat \lambda_\eta$. 
In the case $\hat \lambda_\eta$ is larger than the length of the current segment, the calibration set will include the current segment. 
Given $m_t$ the active set cardinality, the active set is equal to:
\begin{align*}
	\mathcal I^{active} = \lbrace t-m_t+1, \ldots, t\rbrace
\end{align*}

\begin{algorithm}[tb]
	\begin{algorithmic}[1]
		\If{$\ell_t < \hat \ell_\eta$}
		\State{$m_t \gets \ell_t$}
		\Else
		\State{$m_t \gets min(\hat\lambda_\eta, \ell_t)$}
		\EndIf\\
		\Return $m_t$
	\end{algorithmic}
	\caption{Computation of active set cardinality.}
	\label{alg:active-set}
\end{algorithm}

The rest of this section will clarify and prove the claims. The first step is to define what is meant by a ``correct estimate''.
Starting from the observation that in an online context, where quantities are estimated knowing only part of the time series, it is impossible to estimate the quantity more accurately than knowing the whole time series. 
For each data point $X_u$, the oracle status is introduced, noted as $\tilde d_u$. This is the status that $X_u$ would have been given by BKAD, assuming that the breakpoint locations and score values were estimated with knowledge of the entire time series. 

\begin{definition}[Oracle status]\label{def:oracle-status}
	The oracle status, noted $\tilde d_u$, is the status of $X_u$ under the hypothesis that the entire time series is known. Therefore, the breakpoint locations are estimated using the entire time series, and the atypicity score values are estimated using the entire segments. With $T$, the length of the full time series (potentially infinite).
	\begin{align}
		s_{u, T} &= a(X_u, \lbrace X_{\hat\tau_{i}(T)},\ldots,X_{\hat\tau_{i+1}(T)} \rbrace)\\
		\tilde d_u &= \mathbb{1}\left[ \hat p_e(s_{u,T}, S_u)<\varepsilon(\hat p_e(s_{u,T}), \ldots, \hat p_e(s_{u+m,T}) ) \right]
	\end{align}
\end{definition}

The oracle status allows to define the the confidence score associated with an estimated status. It is the probability that the estimated status is the same as the oracle status. A ``correct estimate'' is a status associated to a high confidence score. The ``correct way'' to build the active set is to ensure high confidence score on status. This is ensured by Theorem~\ref{thm:change-vs-oracle-prob} for Algorithm~\ref{alg:active-set}.
\begin{definition}[Confidence Score]
	The confidence score $c_{u,t}$ assigns to the decision made for the data point $X_u$, at time $t$, the  probability that it remains the same under the oracle status
	$$c_{u,t} = \PP\left[d_{u,t}=\tilde d_{u}\right]$$
\end{definition}

Now it's time to clarify what is meant by ``good accuracy'' on score and a point ``correctly assigned to a segment''. 
To this end, let's start by recalling how the status of a data point is established and introduce some notations.
As described in Definition~\ref{def:ideal-bkad}, the status of each data point in the current segment is calculated as follows in three steps,
let $\hat b_t$ be the last breakpoint:
\begin{enumerate}
	\item For all $u$ in $\llbracket\hat b_t,t\rrbracket$ compute the atypicity score $s_{u} = \overline a(X_u, \lbrace X_{\hat b_t}, \ldots, X_{t-m})$.
	\item For all $u$ in $\llbracket\hat b_t,t\rrbracket$ compute the $p$-value $p_u = \hat p_e(s_u, \mathcal S_t)$ 
	\item For all $u$ in $\llbracket\hat b_t,t\rrbracket$ compute the status $d_{u,t} = \mathbb{1}[ p_u < \hat \varepsilon(p_{u,t},\ldots, p_{u+m,t})]$.
\end{enumerate}

Various situations can lead to a change in the status of some data point $X_u$. Before describing these situations, it is useful to introduce the following events:
\begin{itemize}
	\item ``The status of data point $X_u$ at step $t$ is different than the oracle status''  
	$$\V_{u,t} = \left\lbrace  d_{u,t} \neq \tilde d_{u} \right\rbrace$$
	\item  ``The segment to which the data point $X_u$ is assigned at time $t$ changes over time''
	\begin{equation}\label{eq:delay}
		\W_{u,t} = \left\lbrace \exists t'> t, \hat\tau(t') \cap \rrbracket\hat b_t,u\rrbracket \neq \emptyset \right\rbrace
	\end{equation}
\end{itemize}

First, if a breakpoint is detected between $\hat b_t$ and $u$, as described by the event $W_{u,t}$, this means that the score associated with $X_u$ has to be computed from a different training set. Similarly, if a breakpoint is detected between $u$ and $u+m$, it means that the data-driven threshold has to be computed from a different subseries. For these reasons, the probability of a point changing its assigned segment $\PP\left[\W_{u+m,t} \right]$ is of interest and if $\PP\left[\W_{u+m,t} \right]$ is low  the point is said to be ``correctly assigned to a segment''. If no breakpoint is detected between $\hat b_t$ and $u+m$, the assignment of points $u$ to $u+m$ remains unchanged. This event is recorded in $\overline W_{\overline u,t}$ with $\overline u = u+m$. Under the condition of the $\overline W_{\overline u,t}$ event, it is possible for the status to be different from the oracle status if the addition of a data point in the current segment modifies the score value:
$s_{u, T} = \overline a(X_u, \lbrace X_{\hat \tau_{i}(T)}, \ldots, X_{\hat \tau_{i+1}(T)})$, with $\llbracket\hat b_t, t-m\rrbracket \subset \llbracket\hat \tau_{i}(T),\hat \tau_{i+1}(T)\rrbracket$. For this reason,  the score is said to be known with a ``good accuracy'' if the probability $\PP\left[\V_{u,t}|\overline{\W_{\overline u,t}} \right]$ is small.
 
The next proposition describes when it is possible to define $\lambda_\eta$ and $\ell_\eta$ to bound the probabilities $\PP\left[\W_{u,t} \right]$ and $\PP\left[\V_{u,t}|\overline{\W_{\overline u,t}} \right]$, as required by Algorithm~\ref{alg:active-set}:

\begin{proposition}[Stationarity]\label{prop:stationarity}
	Let $\eta>0$.
	\begin{itemize}
		\item  Assuming $f_\tau: \lambda \mapsto \PP\left[\W_{t-\lambda,t} \right]$ is decreasing to 0 and does not depend on $t$. 
		
		Then, there exists $\lambda_\eta$ such that:
		\begin{align}
			\forall t\in \llbracket 1, T\rrbracket, \forall u \in \llbracket 1, t\rrbracket, \quad |u-t|\geq \lambda_\eta,\quad  \PP\left[\W_{u,t} \right]\leq  \eta.
		\end{align}
		The smallest value respecting this property is noted $\lambda_\eta^\star$ . 
		\item Assuming $f_d: \ell \mapsto \PP\left[\V_{u,t}|\overline{\W_{\overline u,t}},\ell_t=\ell \right]$ is decreasing to 0 and does not depend on $t$. 
		Then, there exist a segment length $\ell_\eta$ such that:
		\begin{align}
			\forall t\in \llbracket 1, T\rrbracket, \forall u \in \llbracket 1, t\rrbracket, \quad \ell\geq \ell_\eta,\quad  \PP\left[\V_{u,t}|\overline{\W_{\overline u,t}},\ell_t=\ell \right]\leq  \eta.
		\end{align}
		The smallest value of $\ell_\eta$ is noted $\ell_\eta^\star$. 
	\end{itemize}
\end{proposition}
The conclusions of Proposition~\ref{prop:stationarity} follow directly from the definition of convergence to 0.
Before considering the consequences of this proposition in Theorem~\ref{thm:change-vs-oracle-prob}, the validity of the assumptions is discussed.
The function $f_\tau: \lambda \mapsto \PP\left[\W_{t-\lambda,t} \right]$ gives the probability that the segment assigned to $X_{t-\lambda}$ changes as a function of the distance $\lambda$ from the last observation. It is assumed to be decreasing because the probability of missing a breakpoint decreases with the number of points. 
The function $f_d: \ell \mapsto \PP\left[\V_{u,t}|\overline{\W_{\overline u,t}},\ell_t=\ell \right]$ gives the probability of changing the status of a point conditional on the assigned segment remaining unchanged, as a function of the length $\ell$ of the segment. It is assumed to decrease with the number of points inside the segment. 
%
%
%
Assuming that the probabilities $\PP\left[\W_{t-\lambda,t} \right]$ and $\PP\left[\V_{u,t}|\overline{\W_{\overline u,t}},\ell_t=\ell \right]$ do not depend on $t$, it is possible to use the same model for the entire series. Thus, there is no need to recalculate these probabilities for each observation time. 


According Theorem~\ref{thm:change-vs-oracle-prob}, by re-evaluating only the points in the active set, defined by Algorithm~\ref{alg:active-set} using $\lambda_\eta$ and $\ell_\eta$, the final status of a data point will be the same as the oracle status with a high probability.

\begin{theorem}\label{thm:change-vs-oracle-prob}
	Let $\eta$ be the confidence threshold, $\lambda_\eta$ and $\ell_\eta$ are defined as in Proposition~\ref{prop:stationarity}.
	Let $m$ be the integer defined by $m=max(\lambda_\eta, \ell_\eta)$.
	
	It is assumed that:
	\begin{itemize}
		\item the probability to move the latest breakpoint beyond $m$ is lower than $\eta$.
		\begin{align}\label{eq:move-b}
			\PP\left(\exists t' \hat b_{t'}<\hat b_{t} \mbox{ and } |\hat b_{t'}-t'|>m \bigg| |b_t-t|<m\right) \leq \eta
		\end{align}
	
	    \item The probability of changing the segment assignment depends only on $\lambda_{u,t}=t-u$, the distance between $X_u$ and the end of the time series $t$. 
	    This assumption can be used to calculate the probability of changing the segment assignment within the previous segments (segments that are not the current segment):
	    \begin{align}
	    	&\PP\left(\exists t'>t, \hat \tau(t')\cap \rrbracket \hat\tau_i(t),u\rrbracket \neq \emptyset\bigg| t-\hat \tau(t)=\lambda\right) \notag\\
	    	&\qquad= \PP\left(\exists t'>t, \hat \tau(t')\cap \rrbracket \hat b_t,u\rrbracket \neq \emptyset\bigg| t-\hat b_t=\lambda\right)\label{eq:all-current}
	    \end{align}
	\end{itemize}
	
	Then, applying the ideal BKAD as stated in Definition~\ref{def:ideal-bkad} with the parameters $\lambda'=\lambda_\eta$ and $\ell'=\ell_\eta$, for each $u$, the probability that the final status is different than the oracle status is lower than $\eta$:
	\begin{align}
		\PP(d_{u,T} \neq \tilde d_u)\leq \eta
	\end{align}

	Furthermore, introducing the following notation: 
	For all $t$, for all $u$, let $\hat\tau^{:u}(t) = \lbrace \hat \tau_{i}(t),  \hat\tau_{i}(t)<u\rbrace$ and $\hat\tau^{u+q_1:}(t) = \lbrace \hat \tau_{i}(t),  \hat\tau_{i}(t)>u+q_1\rbrace$. 
	
	Assuming that there is a number $q_1$ such that 
	\begin{align}
		\forall u,t, \quad \hat\tau(t)^{:u} \perp \hat\tau^{u+q_1:}(t) \label{eq:breakpoint-indp}
	\end{align}

	Then, the proportion of status that are different between the final status and the oracle status is lower than $\eta$:
	\begin{align}
	    \lim_{t \rightarrow T} \frac{1}{t} \sum_{u=1}^t \mathbb{1}[d_{u,t} \neq \tilde d_u] \leq \eta
	\end{align}
\end{theorem}
The proof of Theorem~\ref{thm:change-vs-oracle-prob} can be found in Appendix~\ref{appendix:proof-oracle-vs-final}.
The results of Theorem~\ref{thm:change-vs-oracle-prob} support the way the active set is built at the beginning of Algorithm~\ref{alg:generalise-sbad}.
Note that the first result applies locally to one data point. The second result applies to the entire time series.
However, the true values of $\lambda_\eta$ and $\ell_\eta$ are not known, so they need to be estimated. This problem is addressed in Appendix~\ref{sec:delay-proba} and~\ref{sec:stable-seg-change}. Furthermore, this result is for an ideal version of BKAD, which limits its scope in practice. The following section discusses the validity of the various hypotheses.

\subsection{Discussion about theoretical hypotheses}\label{sec:hypo-disc}
The previous theoretical results prove that under ideal conditions BKAD allows to detect anomalies with a control on the FDR. Also the strategy consisting in updating only the active set ensures that the final status are the same as knowing the complete time series, with a low proportion of errors.
These ideal conditions cannot be verified in practice. The approach in this paper is as follows: each component of the BKAD detector is studied to find the best parameters. Then, the detector is empirically tested to see under which conditions it succeeds in detecting anomalies with a control on the FDR.
Now, the different assumptions are examined, their validity is discussed, and the properties that the components must verify are deduced.

First, the assumption \ref{assume.segmentation} cannot be verified. It is impossible to ensure that a breakpoint detector estimates the location of all breakpoints with perfect accuracy, even with a delay of $\lambda>0$. To get closer to these working hypotheses, a powerful breakpoint detector is needed. BKAD uses KCP \cite{arlotKernelMultipleChangepoint} because it has several interesting properties: the number of breakpoints is estimated by model selection, it can detect changes in any feature thanks to kernels and it does not require parametric assumptions that are difficult to verify. For more details, see the dedicated Appendix~\ref{sec:breakpoint}.

The first part of \ref{assume.score}, which assumes that there is an iid oracle score, is always verified.  Indeed, for each segment $i$ in $\llbracket 1, D\rrbracket$, the reference distribution is noted $\mathcal P_{0,i}$ and for each normal data point of the indices $u$ in this segment: $\tilde a(X_u, i) = \PP_{X \sim \mathcal{P}_{0,i}}(X\leq X_u)$ follows a uniform distribution $U(0,1)$. However, there is not always uniqueness of such an oracle score.
For example, if the changes occur only in the mean, then by noting $\mu_i$ the mean of the $i$th segment, $|X_u - \mu_i|$ is also an oracle atypicity score that verifies the iid property. 
However, in practice, the oracle score has to be estimated correctly, which can be difficult depending on the oracle score. 
It is not possible to verify the property \ref{assume.score}. Indeed, it is not possible to know the exact value of the oracle atypicity score. To approximate this property, one needs a measure of nonconformity $\overline a$ that converges as quickly as possible to the value of the oracle atypicity score. As described in Appendix~\ref{sec:score}, the non conformity measure must be robust and efficient.
As a consequence of the fact that the estimated atypicity scores are not iid, the scores from different segments cannot be rigorously used to construct the calibration set. To limit this issue, the calibration set is built from the segments with the most similar distribution to the one of the current segment. This mechanism is described in Appendix~\ref{sec:calibration}.

Another assumption of the ideal BKAD in Definition~\ref{def:ideal-bkad} is that the calibration set is built using only normal data points, which requires knowledge of the true labels. In the Algorithm~\ref{alg:generalise-sbad}, this is obtained by using the labels previously estimated by the anomaly detector. This exposes the calibration set to contamination from undetected anomalies. It can also bias the calibration set by incorrectly removing false positives. This may limit the ability of BKAD to control the FDR.
From a theoretical point of view, this leads to a dependency between the calculation of the $p$-value at time $t$ and the state of the data points at the previous time $t-1$, which complicates the analysis.

Finally, in the previous section, the values of $\lambda_\eta^*$ and $\ell_\eta^*$ were obtained from the functions $f_\tau$ and $f_d$. These quantities are needed to reduce the uncertainty of the BKAD estimates. In Appendix~\ref{sec:delay-proba} and~\ref{sec:stable-seg-change}, estimators of $f_\tau$ and $f_d$ are introduced.

The different components of the detector are discussed and described in more detail in Appendix~\ref{sec:supplement-descriptions}. The next section~\ref{sec:emp-reslt} empirically evaluates the performance of the detector.

\section{Empirical Assessment of FDR control}
\label{sec:emp-reslt}
An anomaly detector based on breakpoint detection has been proposed in Section~\ref{sec:method-description}. The core components are described separately in the Appendix. The theoretical results are introduced in Section~\ref{sec:th-results}. In this section, the empirical performance of the whole anomaly detector is assessed.
The experiments are conducted in two steps.
First, the anomaly detector is applied to several synthetic time series.
The flexibility of the detector is evaluated and the roles played by the kernel and the atypicity score are highlighted. 
Then, the anomaly detector is applied by replacing some estimators with true knowledge in order to explore more deeply the reasons for the errors made by the anomaly detector. More details about the experiments can be found in the Supplementary Materials.
Furthermore, in the next section, the anomaly detector is evaluated against alternative anomaly detectors.

An experimental framework is designed to conduct the experiments and to evaluate different aspects of the anomaly detector. The framework described in Section~\ref{sec:exp-frame} is adapted for different time series and anomaly detector parameters. 
\subsection{Experimental framework}\label{sec:exp-frame}
Let's consider a time series generation process and an anomaly detector. The following steps are repeated on different samples of the time series:
\begin{enumerate}
	\item Generate the time series, according to the the first reference distribution $\mathcal P_{0,1}$, the proportion of anomalies $\pi$, the alternative distribution $\mathcal P_{1,1}$ and the transition rule describing how the parameters of the reference distribution will change between two segments.
	\begin{enumerate}
		\item The number $D$ of breakpoints is generated by $Exp(T/\theta)$, where $\theta$ is the average distance between two breakpoints.
		\item The position of the $D$ breakpoints follows $U([1,T])$. In addition to the previous step, this implies that the process of breakpoint positions is a Poisson process.
		\item The rule is applied iteratively to get the reference and alternative distributions for each segment.
		Two types of rules are considered:
		\begin{itemize}
			\item Breakpoint in the mean with a jump size of $\Delta$. For each $i$ in $\llbracket 1, D-1\rrbracket$, let $\mu_i$ be the mean of the reference distribution in the $i$th segment. The mean of a segment is equal to the mean of the previous one shifted with $\Delta$.
			\begin{align}
			\forall i \in \llbracket 1, D-1 \rrbracket,\quad	\mu_{i+1} = \mu_i + \zeta_i\Delta
			\end{align} 
		 With $\zeta_i$, a random variable following the Rademacher distribution and defining the sign of the jump.
			\item Breakpoint in the variance with a jump scale size of $\Delta$. For each $i$ in $\llbracket 1, D-1\rrbracket$, let $\sigma_i$ be the standard deviation of the reference distribution in the $i$th segment. The standard deviation of a segment is equal to the standard deviation of the previous segment multiplied or divided by $\Delta$.
			\begin{align}\label{eq.rule.std}
				\forall i \in \llbracket 1, D-1 \rrbracket,\quad	\sigma_{i+1} = \exp(\zeta_i\ln\Delta/2)*\sigma_i
			\end{align}
			 With $\zeta_i$, a random variable following the Rademacher distribution and defining if the standard deviation is multiplied or divided by $\Delta$.
		\end{itemize}
		\item The position of anomalies are generated by a Bernoulli distribution:
		$A_t \sim Ber(\pi) $
		\item All the values of the time series are computed as follows:
		$$\forall i \in \llbracket 1, D\rrbracket,\quad \forall t \in \llbracket \tau_i, \tau_{i+1}\llbracket,\quad  \begin{cases}
			X_t \sim \mathcal P_{0,i},& \text{if } A_t=0\\
			X_t \sim \mathcal P_{1,i},              & \text{otherwise}
		\end{cases}$$

	\end{enumerate}
	\item Apply the anomaly detector on the generated time series. Three core components need to be defined: 
	\begin{enumerate}
		\item the appropriate kernel to identify the breakpoints using KCP,
		\item the scoring function $a$
		\item and parameters $n$ for the length of the calibration set and $\lambda$ and $\ell$ to define the active set.
	\end{enumerate}
	\item Compare the detections with true anomalies and calculate the proportion of false discoveries and of false negatives.
\end{enumerate}
The two criteria FDR and FNR are estimated as the average of the FDP and of the FNP over all repetitions. 

The experimental framework is used in different experiments:
In Section~\ref{sec:application-synth}, different synthetic time series are tested and analyzed. In Section~\ref{sec:diagnosis-underperf} the causes of underperformances of the anomaly detector are studied. Furthermore, in Section~\ref{sec:evaluation-competitors}, the proposed anomaly
detector is compared to alternative anomaly detectors using various public data collections. Finally, in the Supplementary Material in Section 7, the effect of hyperparameter choice on performance is evaluated.

\subsection{Evaluation on different scenarios}
\label{sec:application-synth}
The goal of this section is to check if the breakpoint based anomaly detector is able to detect anomalies with a controlled FDR considering different scenarios of time series. 
In these different scenarios, the reference distribution, the alternative distribution that generates the anomalies, and the shifts between two segments vary. The parameters of the anomaly detectors must be adapted to each scenario. Special attention should be paid to the choice of the non-conformity measure and the kernel.

The first scenario considers Gaussian time series with breakpoints in the mean and anomalies in the tail of the distribution. This simplest scenario is used as a reference before evaluating a more complex one. In the second scenario, a Student's distribution is considered. The goal is to evaluate BKAD on a heavy-tailed reference distribution. The third scenario considers a Gaussian mixture time series with breakpoints in the mean and anomalies in the center of the distribution between the two Gaussian modes.
In this case, the detector is checked for anomalies that are not in the tail of the distribution. 
The 4th scenario evaluates the detector on heteroscedastic time series, considering Gaussian time series with breakpoints in the mean and variance simultaneously.
In the 5th scenario, Gaussian data with breakpoints in the variance are used to evaluate how the anomaly detector can be applied with changes in the variance, which is a more difficult case study.
The last scenario uses 2D Gaussian time series with breakpoints in the covariance to evaluate the detector on multidimensional data. Indeed, the breakpoint in the covariance ensures that breakpoints and anomalies cannot be detected by applying the anomaly detector to each dimension. 

Anomalies at the tail of the distribution are detected using the z-score. 
In the third scenario, anomalies occur close to the mean of the distribution, and therefore cannot be detected by $z$-scores. Therefore, the kNN score introduced in \cite{ishimtsev2017conformal}, is used. Indeed, in this case, anomalies can be characterized by their distance from other segment data.
For 2D data, the Mahalannobis distance is applied to reflect the distribution of the data. 

For the breakpoint detector, the Gaussian kernel  with bandwidth estimated using the median heuristic is considered, as presented in Appendix~\ref{sec:breakpoint}.  This kernel enables accurate breakpoint detection in scenarios where the variance remains constant.
 In the case of heteroscedasticity in time series, where the variance changes between two segments, the median heuristic may be more limited.
 Indeed, time series will have parts where the variance is very low and parts where it is very high.  The problem is that a kernel may be good at detecting breakpoints in a low variance context, but have difficulty when the variance is high, and vice versa. It was found to be more advantageous to use a kernel written as the sum of two kernels, one for detecting breakpoints when the variance is low and another for detecting breakpoints when the variance is high.
 
 According to preliminary experiments in Sections~\ref{sec:delay-proba} and~\ref{sec:stable-seg-change}, the active set is built using $\hat\lambda=100$ and $\hat\ell=100$. Based on the rules defined in Appendix~\ref{sec:threshold}, Benjamini-Hochberg is applied to the active set with the modified parameter $\alpha'=\frac{\alpha}{1+\frac{1-\alpha}{m\pi}}$. The calibration set is built according to the rules of Appendix~\ref{sec:threshold}, where the value $n$ is chosen equal to $m/\alpha'-1$. Two cases are considered, $\alpha=0.2$ and $\alpha=0.1$. In the case $\alpha=0.2$, then the following values are chosen $\alpha'=0.1$ and $n=999$. In the case $\alpha=0.1$, then $\alpha'=0.05$ and $n=1999$.

\begin{table}[!b]
	\centering
	\begin{tabular}{l l l l}
		\toprule
		Scenario & $\alpha$ & FDR (std) & FNR (std) \\
		\midrule
		Gaussian data + shift in mean & 0.1 & 0.134 (0.117) & 0.123 (0.213) \\
		Gaussian data + shift in mean & 0.2 & 0.242 (0.121) & 0.039 (0.124) \\
		Student data + shift in mean & 0.1 & 0.158 (0.097) & 0.059 (0.196) \\
		Student data + shift in mean & 0.2 & 0.289 (0.129) & 0.026 (0.123) \\
		MoG data + shift in mean & 0.1 & 0.118 (0.102) & 0.246 (0.303) \\
		MoG data + shift in mean & 0.2 & 0.202 (0.152) & 0.137 (0.220) \\
		Gaussian data + shift in mean and var & 0.1 & 0.134 (0.112) & 0.057 (0.125) \\
		Gaussian data + shift in mean and var & 0.2 & 0.253 (0.125) & 0.018 (0.061) \\
		Gaussian data + shift in var & 0.1 & 0.229 (0.173) & 0.298 (0.314) \\
		Gaussian data + shift in var & 0.2 & 0.282 (0.144) & 0.225 (0.260) \\
		Gaussian data + shift in cov & 0.1 & 0.126 (0.092) & 0.054 (0.145) \\
		\bottomrule
	\end{tabular}

\caption{FDR and FNR of BKAD according to the time series generation scenario.}
\label{tab:synthetical-data-results}
\end{table}

The results of the experiment are summarized in Table~\ref{tab:synthetical-data-results}. Each row corresponds to a scenario for which the results of 50 time series have been averaged to estimate the FDR and FNR, with the standard deviation in parentheses. 
In most of the cases, the FNR is close to 0. This is necessary to ensure the FDR control with the modified BH procedure according to Corollary~\ref{corollary-fdr-bkad}. For example, in the Gaussian scenario with $\alpha=0.2$, the FNR is equal to $0.039$, and the FDR remains close to the desired $\alpha$ level. 
However, it is always slightly higher than alpha. In fact, for the Gaussian scenario with mean shift, it is equal to $0.23$ instead of $\alpha=0.20$, as shown in Table~\ref{tab:synthetical-data-results}. The FDR is more or less controlled in most scenarios when the distribution is varied: Student, Gaussian mixture, or the atypicity score: $z$-score, $k$NN.
For shift in variance, the anomaly detection shows a poor accuracy, since there is a lack of control of the FDR and the FNR is very high. This shows that it is more difficult to detect a breakpoint in the variance than in the mean. The variance of the FDP and FNP is high, of the order of 0.1, which is due to the fact that the result of Theorem~\ref{thm:ideal-detector} is asymptotic and does not guarantee control over series of finite length (here 3000 points).

\subsection{Diagnose the causes of underperformance}
\label{sec:diagnosis-underperf}
Our breakpoint based anomaly detector has been tested on different time series data in Section~\ref{sec:application-synth}, it shows good performances to ensure low FNR with an almost controlled FDR in different cases. However, the FDR is never completely under control, and is always slightly higher than expected. This section examines why this lack of complete control of the FDR occurs by replacing some estimators with knowledge of the true values and evaluate the effect on the FDR. 
The BKAD is applied to the synthetic time series, where some estimators are replaced by true knowledge, called oracle version. Three estimators are chosen to be replaced by their oracle versions:
\begin{itemize}
	\item The breakpoint estimator: can be replaced by the true breakpoint position,
	\item The mean and standard deviation estimators: can be replaced by their true values,
	\item The anomaly removed: As described in Appendix~\ref{sec:calibration}  in the main article when building the calibration set, estimated anomalies are removed to avoid biasing the estimation of the $p$-values. The oracle version of this is to remove the true anomalies. 
\end{itemize} 

 The performances of the different detectors are evaluated on a different laws generating the time series (Student, Gaussian, Mixture of Gaussians noted MoG). 
 Table~\ref{tab:synthetic-oracle-gaussian} gives a synthetic view of the results for Gaussian data.
 Each row represents a different detector. For each component, if the quantity is estimated, it is marked with ``E'', if the true value is used, it is marked with ``O''. Statistically significant differences according to a permutation test \cite{fisher1960design} are marked in bold.
 The complete empirical results can be found in Section 6 of the Supplementary Material. 
 
\begin{table}[!t]
	\centering
	\begin{tabular}{llm{2cm}m{2cm}rr}
		\toprule
		  $\alpha$ & Breakpoint & Mean-and-variance & Anomaly removing &   FDR &   FNR \\
		\midrule
		 0.1& E& E& E     & \textbf{0.134} & 0.123 \\
			& E& E& O     & 0.104          & 0.105 \\
			& O& E& O     & 0.100          & 0.091\\
			& O& O& E     & \textbf{0.165} & 0.041 \\
			& O& O& O     & 0.121          & 0.048 \\\hline
		 0.2& E& E& E     &\textbf{ 0.242} & 0.039 \\
		    & E& E& O     & 0.182          & 0.054 \\
		    & O& E& O     & 0.198          & 0.032 \\
	        & O& O& E     & \textbf{0.301} & 0.018 \\
		    & O& O& O     & 0.197          & 0.018 \\
		\bottomrule
	\end{tabular}
	\caption{Anomaly detector performances with (``O'') and without (``E'') knowledge of true breakpoints, true mean and variance or true anomalies for removing anomalies on Gaussian data.}
	\label{tab:synthetic-oracle-gaussian}
\end{table}
 
 According Table~\ref{tab:synthetic-oracle-gaussian}, 
 it is clear that the control of the FDR is worse when the calibration set is built based on detected anomalies. Indeed, the false positives and false negatives detected at time $t$  will badly affect the detection at time $t + 1$. Despite the fact that a robust score is chosen, these observations lead to a conclusion that the $p$-value estimator is sensitive to:
\begin{itemize}
	 \item False negatives: If there is a missed anomaly in the calibration set, the $p$-values of all data points in the active set will be underestimated. This situation leads to generate more false negatives, which will confound the calibration sets of subsequent instants.
	\item False positives: The $p$-value estimator is also sensitive to false positives due to the way the calibration set is constructed. As a reminder, detected anomalies are replaced by a random points belonging to a segment similar to the current segment. The problem arises when an anomaly is falsely detected.
	Generally speaking a false positive is a point with a high score.  When a false positive is replaced with a random point, its score will be statistically lower.
	Thus, removing the false positives from the calibration set reduces the average score in the calibration set and consequently reduces the $p$-values of the data points in the active set. This leads to more false positives, which will affect the construction of calibration sets at later times.
\end{itemize}
On the other hand, estimating the position of breakpoints or the mean and variance have little impact on the detector's performance.

The conclusion of this analysis is that most of the underperformance relative to the ideal case, such as higher than expected FDR, is explained by the non-robustness of the empirical $p$-value estimator and the contamination of the calibration set by false negatives and false positives.

\section{Evaluation against competitors}
\label{sec:evaluation-competitors}
After studying the conditions that must be met to ensure high detection performance and control of the FDR in the previous sections, the breakpoint detection based anomaly detector (BKAD) proposed in this paper is compared to alternative anomaly detectors from the literature on different data collections. The goal is to determine if and under which conditions the new anomaly detector can improve the state of the art.

\subsection{Methods}
BKAD is evaluated against state-of-the-art anomaly detectors presented in the review \cite{SchmidlAnomalydetectiontime2022}. The most representative   unsupervised anomaly detectors
for univariate time series data are selected. The implementation of \cite{SchmidlAnomalydetectiontime2022} is used, with default hyperparameters.
The detectors selected are those that theoretically capable of detecting anomalies in piecewise iid data. These algorithms fall into two categories: the one that build a context such as a segment, a sliding window or a cluster, and on the other that use subseries instead of single points. On the other hand, predictive or regression models are of little interest on piecewise iid data.

\paragraph{Median \cite{basu2007automatic}}
A sliding windows is used to estimate the median and dispersion parameter of last observations. The atypicity score used is the $z$-score. The main difference with the BKAD approach is the use of sliding windows instead of using a breakpoint detector to define the segments.

\paragraph{CBLOF \cite{he2003discovering}}
Cluster based local outlier factor identifies the cluster to which individual points belong, then it computes the local outlier factor associated with that cluster. The use of clusters is similar to that of breakpoints in that it attempts to group similar points together, but has no temporal notion.

\paragraph{Sub. LOF \cite{breunig2000lof}}
The method divides the time series in subsequences and uses Local Outlier Factor on the subsequences set.

\paragraph{LOF \cite{breunig2000lof}}
The method applies Local Outlier Factor to punctual data. Interesting to compare with ``Sub. LOF'' and ``CBLOF''.

\paragraph{Sub. IF \cite{liu2008isolation}}
The method divides the time series in subsequences and uses Isolation Forest on the subsequences set.

\paragraph{DWT \cite{thill2017time}}
Method based on wavelet to remove noise. Atypicity score is computed using the Gaussian distribution on the Discrete Wavelet Transform, with Haar wavelet. Anomalies can be detected as abnormal Haar coefficients.

\paragraph{FFT \cite{rasheed2009fourier}}
Method based on Fast Fourier Transform. It uses Local outlier factor on the Fast Fourier Transform of the subsequences. Anomalies can be detected as abnormal frequency coefficients.  

\subsection{Threshold}\label{sec:th-bench}
After applying these different methods, an atypicity score is obtained. This score is sufficient to compute the AUC metric, but does not allow detection and calculation of the FDR and FNR without thresholds. To calculate these thresholds, the method introduced in \cite{kronert2023fdr} is used, which guarantees FDR control at a fixed $\alpha$ level in case the time series of scores is iid. The threshold of BKAD is chosen as described in Appendix~\ref{sec:threshold}. Here $\alpha$ is set to $0.2$ for all detectors and time series.

\subsection{Data}
To ensure a comprehensive analysis, different kind of time series data are considered:
\begin{itemize}
	\item Time series with breakpoints
	\item Time series with seasonality
	\item Residual from time-series with seasonality
	\item Real data time series
\end{itemize}

\paragraph{Time series with breakpoints}
The time series with breakpoints are generated according to the experimental design presented in Section~\ref{sec:exp-frame} with the following hyperparameters:
the reference distribution is Gaussian $\mathcal P_{0,1}=\mathcal{N}(0,1)$ and all anomalies follow the law of $4\zeta$, where zeta follows the Rademacher distribution. Anomalies are generated with a proportion of $\pi=0.01$. The breakpoint positions are generated according to the Poisson process with an average segment length of 125. To avoid having too few segments, breakpoints are removed if a segment has less than 100 points. For the benchmark \emph{breakpoint-mean}, breakpoints occur in the mean with a $\Delta=2$. And, for the benchmark \emph{breakpoint-var}, breakpoints occur in the variance with $\Delta=1.5$.

\paragraph{Time series with seasonality}
To study how the anomaly detector behaves on time series not following the statistical model introduced in Section~\ref{sec:modelisation-problem}, time series with seasonality and trend are considered.

Let the following components be given:
\begin{enumerate}
	\item $R_t  \sim \mathcal N(0, \sigma)$, the residual, $\sigma=1$
	\item $A_t \sim B(\pi)$ the abnormality variable, $\pi=0.01$
	\item $S_{1,t} = A_1\sin(2\pi f_1 t)$ the seasonality with long period, where the amplitude $A_1$ and the frequency $f_1$ are random variables, $A_1\in \lbrace 1, 3, 5 \rbrace$ and $f_1 \in \lbrace 5, 10, 20\rbrace$
	\item $S_{2,t} = a_{21}A_1\sin(2\pi w_{21}f_1 t)$ the seasonality with short period, where the frequency multiple $w_{21}$ and the amplitude attenuation are random variable, $a_{21}\in \lbrace 0.5, 0.3, 0.1 \rbrace$ and $w_{21} \in \lbrace 2, 3, 5 \rbrace$
	\item $\sigma_t = \sin(t)+1.5$ the seasonal variance
	\item $T_t=Bt$ the linear trend
\end{enumerate}

The following collections are generated:
\begin{enumerate}
	\item \textit{simple-seasonality}: $X_t = S_{1,t} + (1-A_t)R_t + A_t\zeta_t\Delta'$
	\item  \textit{complex-seasonality}: $X_t = S_{1,t} + S_{2,t} + (1-A_t)R_t + A_t\zeta_t\Delta'$
	\item \textit{variance-seasonality}: $ X_t = ((1-A_t)R_t + A_t\zeta_t\Delta')\sigma_t $
	\item \textit{trend-seasonality}: $X_t = T_t + S_t + (1-A_t)R_t + A_t\zeta_t\Delta'$ 
\end{enumerate}

\paragraph{Residual from time series with seasonality}
In practical applications, to simplify the detection of anomalies, seasonality and other predictable patterns are removed during a preprocessing step.
To evaluate how the anomaly detector performances are affected by this preprocessing step, a new benchmark is built from the residuals extracted for each time series in the ``Time series with seasonality'' benchmark.
In this experiment, the residual is extracted by removing the trend and the seasonality using the ``seasonal\_decompose'' function from the Python library \emph{statsmodels}.

\paragraph{Time series from real data:}

The anomaly detectors are evaluated on various time series datasets coming from different sources. The Numenta Anomaly Benchmark (NAB) from \cite{lavin2015evaluating}, the dodger dataset from the UCI at \cite{ihler2006adaptive} and Mars Science Laboratory (MSL) and Soil Moisture Active Passive (SMAP) provided by NASA in \cite{hundmanDetectingSpacecraftAnomalies2018} are used to build the complete benchmark.

\subsection{Metrics}
To measure the performances of different anomaly detectors, two metrics are reported: The Area Under Curve (AUC)\cite{bradley1997use, hanley1982meaning} in Table~\ref{tab:bench-auc} and the FDR/FNR in Table~\ref{tab:bench-fdpfnp}. 
The advantage of the Area Under the Curve (AUC) is to be able to evaluate the anomaly detector without evaluating the threshold selection method. However, this can also be a limitation for real use, since a threshold is needed for practical applications. 
To determine the ability of anomaly detectors to control the false positive rate to a desired level while keeping the false negative rate low, the FDR and FNR metrics are reported.
The disadvantage of these metrics is that it can be difficult to compare two detectors if one performs better on FDR and the other on FNR. Furthermore, they only take into account values for a single threshold, which have to be precised for detectors that return only an atypicity score. The threshold policy used for this experiment is the one implemented in \cite{kronert2023fdr}, as stated in Section~\ref{sec:th-bench}.

\subsection{Results and analysis}
The results are summarized in two tables. Table~\ref{tab:bench-auc} represents the AUC metric according to benchmarks and anomaly detectors and  Table~\ref{tab:bench-fdpfnp} represents the FDR and FNR metrics.

\begin{table}[tb]
	\centering
	\begin{adjustbox}{max width=\textwidth}
		\begin{tabular}{lccccccccc}
			\toprule
			{Benchmark} & {BKAD (Ours)} & {Median} & {CBLOF} & {Sub. LOF} & {LOF} & {Sub. IF} & {DWT} & {FFT} \\
			\midrule
			{Breakpoint in mean} & \textbf{1.00} & 0.95 & 0.80 & 0.65 & 0.70 & 0.64 & 0.61 & 0.83 \\
			{Breakpoint in variance} & \textbf{0.98} & 0.89 & 0.78 & 0.60 & 0.56 & 0.52 & 0.54 & 0.16 \\
			\midrule
			{Simple seasonality} & 0.88 & \textbf{0.98} & 0.73 & 0.71 & 0.68 & 0.56 & 0.57 & 0.72 \\
			{Complex seasonality} & 0.94 & \textbf{0.98} & 0.85 & 0.72 & 0.79 & 0.57 & 0.55 & 0.64 \\
			{Seasonality in variance} & \textbf{1.00} & 0.99 & 0.92 & 0.56 & 0.87 & 0.54 & 0.57 & 0.71 \\
			{Seasonality and trend} & 0.88 & \textbf{0.98} & 0.67 & 0.71 & 0.63 & 0.53 & 0.57 & 0.72 \\
			\midrule
			{Res. simple seasonality} & \textbf{0.99} & 0.98 & 0.99 & 0.69 & 0.92 & 0.62 & 0.57 & 0.89 \\
			{Res. complex seasonality} & \textbf{1.00} & 0.99 & 1.00 & 0.71 & 0.94 & 0.63 & 0.56 & 0.91 \\
			{Res. seasonality and trend} & \textbf{0.99} & 0.98 & 0.99 & 0.69 & 0.93 & 0.61 & 0.56 & 0.89 \\
			\midrule
			{DODGER} & 0.56 & 0.30 & 0.48 & 0.54 & 0.51 & \textbf{0.67} & 0.65 & 0.30 \\
			{NAB} & 0.57 & 0.45 & 0.54 & 0.67 & 0.48 & 0.66 & \textbf{0.73} & 0.20 \\
			{NASA-MSL} & 0.57 & 0.56 & 0.68 & 0.61 & 0.56 & \textbf{0.84} & 0.81 & 0.56 \\
			{NASA-SMAP} & 0.60 & 0.39 & 0.61 & 0.69 & 0.51 & 0.83 & \textbf{0.90} & 0.47 \\
			\bottomrule
		\end{tabular}
	\end{adjustbox}
	\captionsetup{justification=centering, labelfont=bf, font=small}
	\caption{AUC metric according to the anomaly detectors on benchmarks.}\label{tab:bench-auc}
\end{table}

\begin{table}[tb]
	\begin{adjustbox}{max width=\textwidth}
		\begin{tabular}{lrrrrrrrrrrrrrrrr}
			\toprule
			& \multicolumn{2}{l}{BKAD(Ours)} & \multicolumn{2}{l}{Median} & \multicolumn{2}{l}{CBLOF} & \multicolumn{2}{l}{Sub. LOF} & \multicolumn{2}{l}{LOF} & \multicolumn{2}{l}{Sub. IF} & \multicolumn{2}{l}{DWT} & \multicolumn{2}{l}{FFT} \\
			Benchmark &   FDR &   FNR &           FDR &   FNR &   FDR &   FNR &   FDR &   FNR &   FDR &   FNR &   FDR &   FNR &   FDR &   FNR &   FDR &   FNR \\
			\midrule
			Breakpoint in mean                  & \textbf{0.25} & \textbf{0.04} &         0.07 & 0.48 & 0.83 & 0.52 &           0.99 & 0.81 & 0.52 & 0.70 &          0.99 & 0.44 &     0.99 & 0.01 & 0.93 & 0.46 \\
			Breakpoint in variance              & \textbf{0.36} & \textbf{0.17} &         \textbf{0.19} & \textbf{0.36} & 0.72 & 0.37 &           0.95 & 0.54 & 0.72 & 0.71 &          0.91 & 0.58 &     0.93 & 0.28 & 0.89 & 0.27 \\\hline
			Simple seasonality                  & 0.34 & 0.47 &         \textbf{0.21} & \textbf{0.45} & 0.62 & 0.72 &           0.99 & 0.73 & 0.48 & 0.76 &          0.99 & 0.96 &     0.97 & 0.30 & 0.92 & 0.58 \\
			Complex seasonality                 & 0.27 & 0.44 &         \textbf{0.19} & \textbf{0.44} & 0.41 & 0.64 &           0.99 & 0.68 & 0.37 & 0.64 &          0.99 & 0.95 &     0.95 & 0.43 & 0.95 & 0.65 \\
			Seasonality and trend               & 0.50 & 0.46 &         \textbf{0.22} & \textbf{0.45} & 0.77 & 0.78 &           0.99 & 0.77 & 0.58 & 0.86 &          0.79 & 0.89 &     0.97 & 0.11 & 0.92 & 0.57 \\
			Seasonality in variance             & 0.34 & 0.35 &         \textbf{0.10} & \textbf{0.23} & 0.32 & 0.48 &           0.99 & 0.86 & 0.33 & 0.50 &          0.95 & 0.94 &     0.95 & 0.20 & 0.93 & 0.63 \\\hline
			Res. simple seasonality             & \textbf{0.26} & \textbf{0.42} &         0.19 & 0.56 & \textbf{0.25} & \textbf{0.41} &           0.99 & 0.77 & 0.39 & 0.70 &          0.99 & 0.94 &     0.95 & 0.35 & 0.93 & 0.61 \\
			Res. complex seasonality            & \textbf{0.17} & \textbf{0.15} &         0.12 & 0.31 & \textbf{0.21} & \textbf{0.14} &           0.99 & 0.72 & 0.58 & 0.80 &          0.99 & 0.92 &     0.97 & 0.32 & 0.93 & 0.45 \\
			Res. seasonality and trend          & \textbf{0.26} & \textbf{0.43} &         0.20 & 0.57 & \textbf{0.29} & \textbf{0.41} &           0.99 & 0.79 & 0.44 & 0.73 &          0.97 & 0.94 &     0.95 & 0.36 & 0.92 & 0.64 \\\hline
			dodger                              & 0.41 & 0.66 &         0.79 & 0.99 & 0.97 & 1.00 &           0.71 & 0.92 & 0.89 & 0.09 &          0.39 & 0.91 &     0.70 & 0.55 & 0.89 & 0.09 \\
			NAB                                 & 0.61 & 0.91 &         0.62 & 0.85 & 0.50 & 0.82 &           0.64 & 0.74 & 0.67 & 0.58 &          0.55 & 0.62 &     0.77 & 0.27 & 0.87 & 0.25 \\
			NASA-MSL                            & 0.78 & 0.91 &         0.62 & 0.84 & 0.45 & 0.72 &           0.62 & 0.82 & 0.71 & 0.72 &          0.49 & 0.70 &     0.65 & 0.42 & 0.68 & 0.49 \\
			NASA-SMAP                           & 0.69 & 0.92 &         0.76 & 0.63 & 0.70 & 0.52 &           0.75 & 0.64 & 0.80 & 0.27 &          0.65 & 0.44 &     0.83 & 0.05 & 0.80 & 0.33 \\
			\bottomrule
		\end{tabular}
	\end{adjustbox}
	\caption{FDR and FNR metrics according to the anomaly detectors on benchmarks, ($\alpha=0.2$).}\label{tab:bench-fdpfnp}
\end{table}

    The BKAD detector gets the highest AUC scores on series with breakpoints (``Breakpoint in mean'' and ``Breakpoint in variance''), as shown in Table~\ref{tab:bench-auc}. 
    It can also be seen that this detector remains efficient even when the time series contain seasonality  (``simple seasonality'', ``complex seasonality'',...). This shows the benefits of splitting the time series into simpler segments based on breakpoints, even if it does not follow the model introduced in Section~\ref{sec:modelisation-problem}.
    The results show the importance of preprocessing the data. Indeed, the performance of the detector increases when it is applied to the residuals of the seasonal series instead of the original  seasonal time series, as shown for ``Res. simple seasonality'' or ``Res. complex seasonality''.
    Nevertheless, Table~\ref{tab:bench-fdpfnp} shows that it can be difficult to obtain control of the FDR and FNR even for the best AUC score. This illustrates that FDR control relies heavily on the (piecewise) iid hypothesis.
    Finally, BKAD is not very efficient on tested real data such as (``DODGER'', ``NAB'', ... ) containing anomalies which do not follow  the formalism introduced in Section~\ref{sec:modelisation-problem}. 
    The most efficient methods: ``Sub. IF'' and ``DWT'', define an atypicity score on subseries instead of data points. An interesting approach for the future might be to find a better preprocessing to apply it to real data and improve the anomaly detection.
			
	\section{Conclusion}
	In this paper, an online anomaly detector has been developed that detects anomalies and controls the FDR at a given level $\alpha$ on piecewise stationary time series. The research was conducted to address three challenges:
	\begin{itemize}
		\item Changes in the reference distribution: the changes are detected using a breakpoint detector. Anomalies are retrieved in each homogeneous segment by defining an atypicity score and a calibration set.
		\item Uncertainty: Due to the online nature of the detection, the abnormality status of the data points is uncertain. The notion of an active set is introduced to collect the data points that need to be re-evaluated since there status is too uncertain.
		\item and control of the FDR: modified Benjamini-Hochberg procedure is applied to the active set to control the FDR on the entire time series.
	\end{itemize}
	The result of our research is a modular anomaly detector where all core components have been studied through theoretical or empirical analysis to optimize their performance.
	The detector has been evaluated on a variety of scenarios to understand its strengths and limitations. 
	It demonstrates state-of-the-art capabilities to detect anomalies on time series presenting a distribution shift.
	The main drawback of our method is that it relies on non-robust estimation of $p$-values. Also, the piecewise stationary hypothesis is often not respected in practice.
	Further work concerns the integration of a robust $p$-value estimator and the development of a preprocessing step to apply the anomaly detector to time series that are not piecewise stationary.
	\section{Acknowledgments}
	The authors would like to thank Cristian Preda, director of the MODAL team at Inria, for valuable discussions. 
	The authors would like to thank the anonymous referees, an Associate
	Editor and the Editor for their constructive comments that improved the
	quality of this paper.
	
	\section{Declaration of Competing Interest}
	The authors declare the following financial interests/personal relationships which may be considered as potential competing interests: Worldline Company; Inria. Worldline has a patent pending in connection with this work.
	
	\bibliographystyle{elsarticle-num.bst} 
	\bibliography{new_bibli.bib}

\begin{thebibliography}{10}
\expandafter\ifx\csname url\endcsname\relax
  \def\url#1{\texttt{#1}}\fi
\expandafter\ifx\csname urlprefix\endcsname\relax\def\urlprefix{URL }\fi
\expandafter\ifx\csname href\endcsname\relax
  \def\href#1#2{#2} \def\path#1{#1}\fi

\bibitem{hawkins1980}
D.~M. Hawkins, Identification of outliers, Vol.~11, Springer, 1980.

\bibitem{PangANDEAAnomalyNovelty2022}
G.~Pang, J.~Li, A.~van~den Hengel, L.~Cao, T.~G. Dietterich, Andea: anomaly and novelty detection, explanation, and accommodation, in: Proceedings of the 28th ACM SIGKDD Conference on Knowledge Discovery and Data Mining, 2022, pp. 4892--4893.

\bibitem{LiSurveyHeartAnomaly2020}
H.~Li, P.~Boulanger, A survey of heart anomaly detection using ambulatory electrocardiogram (ecg), Sensors 20~(5) (2020) 1461.

\bibitem{chandola2009anomaly}
V.~Chandola, A.~Banerjee, V.~Kumar, Anomaly detection: A survey, ACM computing surveys (CSUR) 41~(3) (2009) 1--58.

\bibitem{goldstein2016comparative}
M.~Goldstein, S.~Uchida, A comparative evaluation of unsupervised anomaly detection algorithms for multivariate data, PloS one 11~(4) (2016) e0152173.

\bibitem{blazquez2021review}
A.~Bl{\'a}zquez-Garc{\'\i}a, A.~Conde, U.~Mori, J.~A. Lozano, A review on outlier/anomaly detection in time series data, ACM Computing Surveys (CSUR) 54~(3) (2021) 1--33.

\bibitem{guha2016robust}
S.~Guha, N.~Mishra, G.~Roy, O.~Schrijvers, Robust random cut forest based anomaly detection on streams, in: International conference on machine learning, PMLR, 2016, pp. 2712--2721.

\bibitem{na2018dilof}
G.~S. Na, D.~Kim, H.~Yu, Dilof: Effective and memory efficient local outlier detection in data streams, in: Proceedings of the 24th ACM SIGKDD international conference on knowledge discovery \& data mining, 2018, pp. 1993--2002.

\bibitem{cvach2012monitor}
M.~Cvach, Monitor alarm fatigue: an integrative review, Biomedical instrumentation \& technology 46~(4) (2012) 268--277.

\bibitem{blum2010alarms}
J.~M. Blum, K.~K. Tremper, Alarms in the intensive care unit: too much of a good thing is dangerous: is it time to add some intelligence to alarms?, Critical care medicine 38~(2) (2010) 702--703.

\bibitem{solet2012managing}
J.~M. Solet, P.~R. Barach, Managing alarm fatigue in cardiac care, Progress in Pediatric Cardiology 33~(1) (2012) 85--90.

\bibitem{lewandowska2023determining}
K.~Lewandowska, W.~M{\k{e}}drzycka-D{\k{a}}browska, L.~Tomaszek, M.~Wujtewicz, Determining factors of alarm fatigue among nurses in intensive care units—a polish pilot study, Journal of Clinical Medicine 12~(9) (2023) 3120.

\bibitem{MarySemisupervisedmultipletesting2022}
D.~Mary, E.~Roquain, Semi-supervised multiple testing, Electronic Journal of Statistics 16~(2) (2022) 4926--4981.

\bibitem{marandon2022}
A.~Marandon, L.~Lei, D.~Mary, E.~Roquain, Machine learning meets false discovery rate, arXiv preprint arXiv:2208.06685 (2022).

\bibitem{benjamini1995}
Y.~Benjamini, Y.~Hochberg, Controlling the false discovery rate: a practical and powerful approach to multiple testing, Journal of the Royal statistical society: series B (Methodological) 57~(1) (1995) 289--300.

\bibitem{Benjaminicontrolfalsediscovery2001}
Y.~Benjamini, D.~Yekutieli, The control of the false discovery rate in multiple testing under dependency, Annals of statistics (2001) 1165--1188.

\bibitem{kronert2023fdr}
E.~Kr{\"o}nert, A.~C{\'e}lisse, D.~Hattab, Fdr control for online anomaly detection, arXiv preprint arXiv:2312.01969 (2023).

\bibitem{aleneziDenialServiceDetection2013}
M.~Alenezi, M.~J. Reed, Denial of service detection through tcp congestion window analysis, in: World Congress on Internet Security (WorldCIS-2013), IEEE, 2013, pp. 145--150.

\bibitem{fischRealTimeAnomaly2020}
A.~T. Fisch, L.~Bardwell, I.~A. Eckley, Real time anomaly detection and categorisation, Statistics and Computing 32~(4) (2022) 55.

\bibitem{huber1992robust}
P.~J. Huber, Robust estimation of a location parameter, in: Breakthroughs in statistics: Methodology and distribution, Springer, 1992, pp. 492--518.

\bibitem{staerman2022functional}
G.~Staerman, Functional anomaly detection and robust estimation, Ph.D. thesis, Institut polytechnique de Paris (2022).

\bibitem{arlotKernelMultipleChangepoint}
S.~Arlot, A.~Celisse, Z.~Harchaoui, A kernel multiple change-point algorithm via model selection, Journal of machine learning research 20~(162) (2019).

\bibitem{ishimtsev2017conformal}
V.~Ishimtsev, A.~Bernstein, E.~Burnaev, I.~Nazarov, Conformal $ k $-nn anomaly detector for univariate data streams, in: Conformal and Probabilistic Prediction and Applications, PMLR, 2017, pp. 213--227.

\bibitem{fisher1960design}
R.~A. Fisher, et~al., The design of experiments., no. 7th Ed, Oliver and Boyd. London and Edinburgh, 1960.

\bibitem{SchmidlAnomalydetectiontime2022}
S.~Schmidl, P.~Wenig, T.~Papenbrock, Anomaly detection in time series: a comprehensive evaluation, Proceedings of the VLDB Endowment 15~(9) (2022) 1779--1797.

\bibitem{basu2007automatic}
S.~Basu, M.~Meckesheimer, Automatic outlier detection for time series: an application to sensor data, Knowledge and Information Systems 11 (2007) 137--154.

\bibitem{he2003discovering}
Z.~He, X.~Xu, S.~Deng, Discovering cluster-based local outliers, Pattern recognition letters 24~(9-10) (2003) 1641--1650.

\bibitem{breunig2000lof}
M.~M. Breunig, H.-P. Kriegel, R.~T. Ng, J.~Sander, Lof: identifying density-based local outliers, in: Proceedings of the 2000 ACM SIGMOD international conference on Management of data, 2000, pp. 93--104.

\bibitem{liu2008isolation}
F.~T. Liu, K.~M. Ting, Z.-H. Zhou, Isolation forest, in: 2008 eighth ieee international conference on data mining, IEEE, 2008, pp. 413--422.

\bibitem{thill2017time}
M.~Thill, W.~Konen, T.~B{\"a}ck, Time series anomaly detection with discrete wavelet transforms and maximum likelihood estimation, in: Intern. Conference on Time Series (ITISE), Vol.~2, 2017, pp. 11--23.

\bibitem{rasheed2009fourier}
F.~Rasheed, P.~Peng, R.~Alhajj, J.~Rokne, Fourier transform based spatial outlier mining, in: Intelligent Data Engineering and Automated Learning-IDEAL 2009: 10th International Conference, Burgos, Spain, September 23-26, 2009. Proceedings 10, Springer, 2009, pp. 317--324.

\bibitem{lavin2015evaluating}
A.~Lavin, S.~Ahmad, Evaluating real-time anomaly detection algorithms--the numenta anomaly benchmark, in: 2015 IEEE 14th international conference on machine learning and applications (ICMLA), IEEE, 2015, pp. 38--44.

\bibitem{ihler2006adaptive}
A.~Ihler, J.~Hutchins, P.~Smyth, Adaptive event detection with time-varying poisson processes, in: Proceedings of the 12th ACM SIGKDD international conference on Knowledge discovery and data mining, 2006, pp. 207--216.

\bibitem{hundmanDetectingSpacecraftAnomalies2018}
K.~Hundman, V.~Constantinou, C.~Laporte, I.~Colwell, T.~Soderstrom, Detecting spacecraft anomalies using lstms and nonparametric dynamic thresholding, in: Proceedings of the 24th ACM SIGKDD international conference on knowledge discovery \& data mining, 2018, pp. 387--395.

\bibitem{bradley1997use}
A.~P. Bradley, The use of the area under the roc curve in the evaluation of machine learning algorithms, Pattern recognition 30~(7) (1997) 1145--1159.

\bibitem{hanley1982meaning}
J.~A. Hanley, B.~J. McNeil, The meaning and use of the area under a receiver operating characteristic (roc) curve., Radiology 143~(1) (1982) 29--36.

\bibitem{blum1963strong}
J.~Blum, D.~L. Hanson, L.~H. Koopmans, On the strong law of large numbers for a class of stochastic processes, Sandia Corporation, 1963.

\bibitem{gray2009probability}
R.~M. Gray, R.~Gray, Probability, random processes, and ergodic properties, Vol.~1, Springer, 2009.

\bibitem{csorgő1992law}
S.~Cs{\"o}rg{\"o}, On the law of large numbers for the bootstrap mean, Statistics \& probability letters 14~(1) (1992) 1--7.

\bibitem{truong2020selective}
C.~Truong, L.~Oudre, N.~Vayatis, Selective review of offline change point detection methods, Signal Processing 167 (2020) 107299.

\bibitem{baudry2012slope}
J.-P. Baudry, C.~Maugis, B.~Michel, Slope heuristics: overview and implementation, Statistics and Computing 22 (2012) 455--470.

\bibitem{celisse2018new}
A.~Celisse, G.~Marot, M.~Pierre-Jean, G.~Rigaill, New efficient algorithms for multiple change-point detection with reproducing kernels, Computational Statistics \& Data Analysis 128 (2018) 200--220.

\bibitem{fukumizu2009kernel}
K.~Fukumizu, A.~Gretton, G.~Lanckriet, B.~Sch{\"o}lkopf, B.~K. Sriperumbudur, Kernel choice and classifiability for rkhs embeddings of probability distributions, Advances in neural information processing systems 22 (2009).

\bibitem{garreau2017large}
D.~Garreau, W.~Jitkrittum, M.~Kanagawa, Large sample analysis of the median heuristic, arXiv preprint arXiv:1707.07269 (2017).

\bibitem{shafer2008tutorial}
G.~Shafer, V.~Vovk, A tutorial on conformal prediction., Journal of Machine Learning Research 9~(3) (2008).

\bibitem{staudte2011robust}
R.~G. Staudte, S.~J. Sheather, Robust estimation and testing, John Wiley \& Sons, 2011.

\bibitem{rousseeuw2018anomaly}
P.~J. Rousseeuw, M.~Hubert, Anomaly detection by robust statistics, Wiley Interdisciplinary Reviews: Data Mining and Knowledge Discovery 8~(2) (2018) e1236.

\bibitem{zhou2017anomaly}
C.~Zhou, R.~C. Paffenroth, Anomaly detection with robust deep autoencoders, in: Proceedings of the 23rd ACM SIGKDD international conference on knowledge discovery and data mining, 2017, pp. 665--674.

\bibitem{koyuncu2009efficient}
N.~Koyuncu, C.~Kadilar, Efficient estimators for the population mean, Hacettepe Journal of Mathematics and Statistics 38~(2) (2009) 217--225.

\bibitem{shoemaker1982robust}
L.~H. Shoemaker, T.~P. Hettmansperger, Robust estimates and tests for the one-and two-sample scale models, Biometrika 69~(1) (1982) 47--53.

\bibitem{LaxhammarConformalanomalydetection2014}
R.~Laxhammar, Conformal anomaly detection: Detecting abnormal trajectories in surveillance applications, Ph.D. thesis, University of Sk{\"o}vde (2014).

\bibitem{bhattacharyya1943measure}
A.~Bhattacharyya, On a measure of divergence between two statistical populations defined by their probability distribution, Bulletin of the Calcutta Mathematical Society 35 (1943) 99--110.

\bibitem{storey2004strong}
J.~D. Storey, J.~E. Taylor, D.~Siegmund, Strong control, conservative point estimation and simultaneous conservative consistency of false discovery rates: a unified approach, Journal of the Royal Statistical Society Series B: Statistical Methodology 66~(1) (2004) 187--205.

\end{thebibliography}

	
\begin{appendix}

\section{Proofs}

\subsection{Proof of Theorem~\ref{thm:ideal-detector}}
\label{appendix:proof-bkad-fdr}

\begin{proof}[Proof of Theorem~\ref{thm:ideal-detector}]
	Similar to the proof of Theorems 3 and 4 of \cite{kronert2023fdr}, the FDP is written as the ratio of $R_1^t$ and $FP_1^t$.
	\begin{align}
		FDP_1^\infty = \lim_{t\rightarrow \infty}FDP_1^t = \lim_{t\rightarrow \infty}\frac{FP_1^t}{R_1^t} =  \lim_{t\rightarrow \infty} \frac{\sum_{u=1}^{t}(1-A_u)d_{u,t}}{\sum_{u=1}^{t}d_{u,t}}
	\end{align}
	 In the next part of the proof the numerator and denominator are made to converge so that the mFDR expression appears.
	The main steps of the proof are:
	\begin{enumerate}
		\item First, it is shown that for any $u$, as long as $t$ is large enough, then $s_{u,t}=\tilde a(X_u,i_u)$. 
		\item Then, it is deduced that for any $u$, as long as $t$ is large enough, then $p_{u,t}$ are identically distributed. 
		\item Similarly, it is shown that for any $u$, as long as $t$ is large enough, then $d_{u,t}$ are identically distributed.
		\item Finally, since $d_{u,t}$ is identically distributed and respects a mixing property, an ergodicity theorem allows to conclude that $1/t R_1^t$ converges to $\mathbb E[d_{1,\infty}]$ 
	\end{enumerate}

    Some notations are introduced:
    \begin{itemize}
    	\item $Q_s$ the distribution of $\tilde a(X_u,i_u)$.
    	\item $Q_p$ the distribution of the $p$-value $p=\hat p_e(s, \mathcal S_{cal})$, when $s$ and all elements of $\mathcal S_{cal}$ follow $Q_s$ and are independent.
    	\item $Q_d$ the distribution of the status $d=\mathbb 1[p_1 < \varepsilon(p_1,\ldots, p_m)]$, when all $p$-values $p_1, \ldots, p_m$ are computed according to the same calibration set and follow $Q_p$.
    \end{itemize}

    \textbf{Step 1: The scores are iid distributed for $t$ sufficiently large.}
       
    Let $u\in \llbracket 1, \infty\rrbracket$. $u$ belongs to a unique true segment $i$, delimited by the breakpoints $\tau_i$ and $\tau_{i+1}$. 
    
       When $t>u+\lambda^*$, according to (\ref{assume.segmentation}) $\hat \tau(t) \cap \llbracket 1,u\rrbracket = \tau \cap \llbracket 1,u\rrbracket$, then Eq.~\ref{eq:score-ideal-bkad} gives: 
    \begin{align}
    	s_{u,t} = \overline a(X_u, \lbrace X_{\tau_i},\ldots, X_{min(\tau_i+\ell^*,t)}\rbrace)
    \end{align}
    Furthermore, when $t>u+max(\lambda^*,\ell^*)=u+m$, there are more than $\ell^*$ data points in the segment, then (\ref{assume.score}) gives:
    \begin{align}
    	s_{u,t} =  \overline a(X_u, \lbrace X_{\tau_i},\ldots, X_{\tau_i+\ell^*}\rbrace) = \tilde a(X_u,i)
    \end{align}
    The true score $\tilde a(X_u,i)$ is iid by assumption. For these reasons $s_{u,t}$ follows $Q_s$.
    
    Since the $p$-value is calculated by comparing the score of a data point with the scores from a calibration set, it can be deduced that the $p$-values are identically distributed.
    
    \textbf{Step 2: The $p$-values are identically distributed for $t$ sufficiently large.}
    
    Let $u$ be in $\llbracket 1, T \rrbracket$, according to Eq.~\ref{eq:pvalue-ideal-bkad} the $p$-value is estimated as the following:
    \begin{align}
    	p_{u,t} &= p(s_{u,t}, \mathcal S_{t})\\
    	\mbox{with }\mathcal S_{t} &= \lbrace s_{h(t-m,1),t},\ldots, s_{h(t-m,n),t}\rbrace 
    \end{align}
     By definition of $h$: $h(t-m, j)\leq t-m$, and thus according to the previous paragraph, all elements of $\mathcal S_{t}$ are identically distributed. All $p$-values associated with a score that follows the $\mathcal Q_s$ distribution follow the same distribution. In particular, all $p$-values $p_{u,t}$ follow the same distribution $Q_p$, as soon as $t>u+m$.
    
    The status of a point depends only on the $p$-value of the point and the $p$-values used to calculate the data-driven threshold. It has been shown that the $p$-values follow the same distribution. In the next paragraph, it is deduced that the statuses are also identically distributed. 
    
     \textbf{Step 3: The decision series $(d_{u,t})$ is identically distributed for $t$ sufficiently large.}
     
    Let $i$ be in $\llbracket 1, D\rrbracket$ and defining a segment $\llbracket \tau_i, \tau_{i+1}-1\rrbracket$. The cases are separated according to the position of the point in the real segment: at the end of the segment or in the middle of the segment.
    
    $\qquad$ Case 1: The data point belongs to the end of the segment, $u \in \llbracket \tau_{i+1}-m, \tau_{i+1}-1\rrbracket$. Two steps are required. First it is shown that $d_{u,t}$ verifies the property for $t=u+m$, then it is shown that for any $t>u+m$, $d_{u,t} =d_{u,u+m}$.
    
    First, let $t=\tau_{i+1}+m$. According to (\ref{assume.segmentation}): $\hat \tau(t) \cap \llbracket 1,\tau_{i+1}\rrbracket = \tau \cap \llbracket 1,\tau_{i+1}\rrbracket$. 
    Thus, the current segment at time $t$ is equal to $\llbracket \tau_{i+1},t\rrbracket$ and has exactly $m$ points. Then, since the rule of closing the previous segment from Eq.~\ref{eq:closing-seg-ideal-bkad} is applied, it gives
    \begin{align*}
    	d_{u,t} = \mathbb{1}[\hat p_e(s_{u,t},\mathcal S_{cal,t})< \varepsilon(\hat p_e(s_{\tau_{i+1}-m,t},\mathcal S_{cal,t}),\ldots,\hat p_e(s_{\tau_{i+1}-1,t},\mathcal S_{cal,t}))]
    \end{align*}
    Since all score variables $s_{\tau_{i+1}-m,t}, \ldots, s_{\tau_{i+1}-1,t}$ follow the distribution $\mathcal Q_s$, according to (\ref{assume.score}), and all $p$-values are computed using the same calibration set $\mathcal S_{cal,t}$, then $d_{u,t}$ follows $\mathcal Q_d$.
    
    Then, for $t>\tau_{i+1}+m$, $d_{u,t}=d_{u,t-1}$. This implies that the limit value $d_{u,\infty}$ is equal to $d_{u,\tau_{i+1}+m}$ which follows the distribution $\mathcal Q_{d}$.
    
    \medskip
    
     $\qquad$ Case 2: The data point belongs to the middle of the segment, $u \in \llbracket \tau_i, \tau_{i+1}-m\rrbracket$. 
     
     To prove that $d_{u,t}$ follows the distribution $\mathcal Q_d$ for all $t\geq u+m$, we have three steps. First, it is shown that the property holds for $t=u+m$. Then it is shown that the status $d_{u,t}$ may possibly be modified for $t$ in $\rrbracket u+m, u+2m\rrbracket$, but that $d_{u,t}$ always follows the $\mathcal Q_d$ distribution. Finally, it is verified that $d_{u,t}$ is constant from $t>u+2m$.
    
   First, let $t=u+m$. Although by hypothesis $\llbracket u,u+m\rrbracket$ contains no true breakpoints, the $\hat \tau_t$ detector can detect a false breakpoint. Cases are separated according to whether a breakpoint was detected or not.
    \begin{itemize}
    	\item Case 2a, there is no (false) breakpoint in $\llbracket u, u+m\rrbracket$. Thus, the current segment is $\llbracket\tau_i, u+m\rrbracket$, according Eq.~\ref{eq:end-seg-ideal-bkad}. 
    	\begin{align}
    		d_{u,t} = \mathbb{1}[\hat p_e(s_{u,t},\mathcal S_{cal,t})< \varepsilon(\hat p_e(s_{u,t},\mathcal S_{cal,t}),\ldots,\hat p_e(s_{u+m,t},\mathcal S_{cal,t}))]
    	\end{align}
        $d_{u,t}$ follows $\mathcal Q_d$ 
        \item Case 2b, there is a (false) breakpoint in $\llbracket u, u+m\rrbracket$, this breakpoint is noted $\hat b_t$.
        The current segment $\llbracket \hat b_t,u+m\rrbracket$ contains less than $m$ points. Thus, according to Eq.~\ref{eq:closing-seg-ideal-bkad}, $d_{u,t}$ takes the value:
        \begin{align}
        	d_{u,t} = \mathbb{1}[\hat p_e(s_{u,t},\mathcal S_{cal,t})< \varepsilon(\hat p_e(s_{\hat b_t-m,t},\mathcal S_{cal,t}),\ldots,\hat p_e(s_{\hat b_t-1,t},\mathcal S_{cal,t}))]
        \end{align} 
        $d_{u,t}$ follows $\mathcal Q_{d}$.
    \end{itemize}
     When $t=u+m$, in both cases $d_{u,t}$ follows the distribution $Q_d$. Next, check that this property remains true even when $t$ is greater than $u+m$.
     
     Then, for $t$ in $\rrbracket u+m, u+2m\rrbracket$, the cases are split again according to the detection of a breakpoint in $\llbracket u, u+m\rrbracket$:
     \begin{itemize}
     	\item Case 2a': There are no detected breakpoint in $\llbracket u, u+m\rrbracket$. According Eq.~\ref{eq:continu-ideal-bkad}, in this case the status is not updated and $d_{u,t} = d_{u,t-1}$.
     	\item Case 2b': A (false) breakpoint is detected in $\llbracket t-m, u+m\rrbracket\subset\llbracket u, u+m\rrbracket$ and noted $\hat b_t$. The current segment $\llbracket\hat b_t,u+m\rrbracket$ contains less than $m$ points, $d_{u,t}$ is updated according to the rule:
     	\begin{align}
     		d_{u,t} = \mathbb{1}[\hat p_e(s_{u,t},\mathcal S_{cal,t})< \varepsilon(\hat p_e(s_{\hat b_t-m,t},\mathcal S_{cal,t}),\ldots,\hat p_e(s_{\hat b_t-1,t},\mathcal S_{cal,t}))]
     	\end{align}
         Once again, $d_{u,t}$ follows the $\mathcal Q_d$ distribution.
     \end{itemize}
 
     Finally, for $t>u+2m$, assuming (\ref{assume.segmentation}), there is no breakpoint in $\llbracket u,u+m\rrbracket$ and therefore: $d_{u,t}=d_{u,t-1}$
 
      By induction $d_{u,t}$ follows the law $\mathcal Q_d$ as soon as $t$ is greater than $u+m$.
      Therefore the series of final decisions $d_{u,\infty}$ is identically distributed and follows the law $\mathcal Q_d$.
    
      \textbf{Step 4: Numerator, denominator and ratio convergence}
      After having proved that the status series $d_{u,\infty}$ is identically distributed, the following shows that the empirical mean of the series converges to its expectation.

   It was shown in the previous step that $d_{u,\infty}$ is identically distributed. Furthermore, $d_{u_1,\infty}\perp d_{u_2,\infty}$ if $|u_1 - u_2|>m+n$. Using the corollary of Theorem 3 in \cite{blum1963strong} this gives almost surely convergence:
    \begin{align}
    	\lim_{t\rightarrow \infty}\frac{1}{t}R_1^t &=\lim_{t\rightarrow \infty}\frac{1}{t}\sum_{u=1}^t d_{u,\infty} \\ &=\mathbb{E}[d_{1,\infty}]
    \end{align} 
    Similarly, it gives the almost surely convergence of the numerator $FP_1^t$.
    \begin{align}
    	\lim_{t\rightarrow \infty}\frac{1}{t}FP_1^t &=
    	\lim_{t\rightarrow \infty}\frac{1}{t}\sum_{u=1}^t A_ud_{u,\infty} \\
    	&= \mathbb{E}[A_1d_{1,\infty}]
    \end{align}
     
     Since both the numerator and the denominator converge almost surely, this leads to the almost sure convergence of the ratio which is the FDP:
      \begin{align}
     	\lim_{t\rightarrow \infty}FDP_t = \frac{\mathbb{E}[A_1d_{1,\infty}]}{\mathbb{E}[d_{1,\infty}]}
     \end{align} 
     
    \textbf{Step 5: Link with mFDR}
    So far it has been proved that the FDP of the complete series converges to the ratio of the expectation of $d_{1,\infty}$ and $A_1d_{1,\infty}$. In the following, this ratio is linked to the mFDR computed on a subseries of size $m$.
    
     This result comes from the permutation invariance of $\varepsilon(\cdot)$ which gives
     \begin{align*}
     	\mathbb E\left[\sum_{u=1}^m \mathbb{1}[\hat p_{u,t} < \varepsilon(\hat p_{1,t}, \ldots, \hat p_{1,t})]\right] &= \sum_{u=1}^m\mathbb E\left[ \mathbb{1}[\hat p_{u,t} < \varepsilon(\hat p_{1,t}, \ldots, \hat p_{1,t})]\right]\\
     	&= m\mathbb E\left[ \mathbb{1}[\hat p_{1,t} < \varepsilon(\hat p_{1,t},  \ldots, \hat p_{1,t})]\right]\\
     	&= m\mathbb E\left[ d_{1,t}\right]
     \end{align*}
      This implies that the FDP limit can be written as mFDR:
       \begin{align}
       	\lim_{t\rightarrow \infty}FDP_t = \frac{\mathbb{E}[A_1d_{1,\infty}] \times m}{\mathbb{E}[d_{1,\infty}] \times m} =  \frac{\mathbb{E}[FP_1^m]}{\mathbb{E}[R_1^m]} = mFDR_1^m
       \end{align} 
    
\end{proof}
\subsection{Proof of Corollary \ref{corollary-fdr-bkad}}
\begin{proof}[Proof of Corollary \ref{corollary-fdr-bkad}]
	According to Theorem~\ref{thm:ideal-detector}, the FDP of the complete time series is equal to the mFDR of the subseries:
	\begin{align*}
		\lim_{t \rightarrow \infty}FDP_1^t = mFDR_1^m
	\end{align*} 
	
	The following steps of the proof show that $mFDR_1^m = (1-\pi)\alpha$ using various results of \cite{kronert2023fdr}.
	
	First, according to Proposition~3 in \cite{kronert2023fdr}, using BH the $mFDR_1^m$ of the subseries can be expressed using the number of rejections and the FDR on the subseries.
	\begin{align}
		mFDR_1^m = \frac{\mathbb{E}[R_1^m(1)]}{\mathbb{E}[R_1^m]}FDR_1^m\label{eq:mfdr_fdr_proof}
	\end{align} 
	From the assumptions expressed in Eq.~\ref{eq:heuristic.power}, it follows that
	\begin{align*}
		\frac{\mathbb{E}[R_1^m(1)]}{\mathbb{E}[R_1^m]} = 1 + \frac{1-\alpha}{m\pi} 
	\end{align*}
	As described in Definition~\ref{def:ideal-bkpt}, all $p$-values are calculated from a single calibration set, furthermore, the cardinality of the calibration set is equal to $n=\ell m/\alpha'-1$. Then according to Theorem 3.4 in  \cite{marandon2022} and Corollary 3 in \cite{kronert2023fdr}: 
	\begin{align*}
		FDR_1^m = (1-\pi)\alpha' = \left(1 + \frac{1-\alpha}{m\pi} \right)^{-1}(1-\pi)\alpha
	\end{align*}
	
	Injecting the value of $FDR_1^m$ and $\frac{\mathbb{E}[R_1^m(1)]}{\mathbb{E}[R_1^m}$ in Eq.~\ref{eq:mfdr_fdr_proof}, it gives:
	\begin{align}
		mFDR_1^m &= \left(1 + \frac{1-\alpha}{m\pi} \right)\left(1 + \frac{1-\alpha}{m\pi} \right)^{-1}(1-\pi)\alpha\\
		&= (1-\pi)\alpha
	\end{align}
	This result completes the proof.
\end{proof}

\subsection{Proof of Proposition~\ref{prop:stationarity}}
\begin{proof}[Proof of Proposition~\ref{prop:stationarity}]
	By assumption the probability $\PP\left[\W_{t-\lambda,t}\right]$ does not depend on $t$ and is noted $f_\tau(\lambda)$. According to the second assumption, $f_\tau(\lambda)$ converges to $0$ when $\lambda$ tends to $+\infty$. Therefore, by definition of convergence:
	\begin{align*}
		\forall \eta>0, \exists \lambda_\eta>0 \quad \lambda\geq, \forall \lambda_\eta,\quad  f_\tau(\lambda)\leq  \eta.
	\end{align*}
	Moreover, by definition $\lambda = t-u$, it follows that:
	\begin{align*}
		\forall \eta>0,\quad \exists \lambda>0, \quad\forall t\in \llbracket 1, T\rrbracket, \forall u \in \llbracket 1, t\rrbracket, \quad |u-t|\geq \lambda_\eta,\quad  \PP\left[\W_{u,t} \right]\leq  \eta.
	\end{align*}
	The second result is proven using similar arguments.
\end{proof}

\subsection{Proof of Theorem \ref{thm:change-vs-oracle-prob}}\label{appendix:proof-oracle-vs-final}
\begin{proof}[Proof of Theorem~\ref{thm:change-vs-oracle-prob}]
The two statements are proven separately.
First, it is shown that for every $\eta$, the probability that the final status differs from the oracle status is less than $3\eta$. Then, a mixing property allows to prove that the proportion of differences along the time series is less than $3\eta$.
	
	\textbf{Proof of the first statement:}
	To prove the first statement, two steps are taken: first, the final decision about the status of $X_u$ is characterized. According to the way BKAD works, as described in Definition~\ref{def:ideal-bkad}, there are two possibilities: either the final decision is taken when $u$ belongs to the current segment and is not updated thereafter, as stated in Eq.~\ref{eq:continu-ideal-bkad}. Or the status of $X_u$ is updated when the segment to which $u$ belongs is closed, as stated in Eq.~\ref{eq:closing-seg-ideal-bkad}.
	Once the final status has been characterized, the probability that it differs from the oracle status is calculated.
	
    Let $u$ be in $\llbracket 1, T \rrbracket$.
	
	In case there is $t$ such that: $|b_t-t|<m$ and $u \in \llbracket \hat b_t-m,\hat b_t \rrbracket$. Let $t'$ be the largest one. This corresponds to the case where the status of $X_u$ is updated after detecting a breakpoint which closes the segment to which $u$ is assigned.
	According to Eq.~\ref{eq:closing-seg-ideal-bkad}
	\begin{align}
		d_{u,t'} = \mathbb{1}[\hat p_{u,t'}<\varepsilon(\hat p_{u,b_{t'}-m},\ldots,p_{u,b_{t'}} )]
		\end{align}
	
	Otherwise, let $t'=u+m$. This corresponds to the case where the final status associated with $X_u$ is taken at the time it belongs to the current segment.
	According to Eq.~\ref{eq:continu-ideal-bkad}
	\begin{align}
		d_{u,t'} = \mathbb{1}[\hat p_{u,t'}<\varepsilon(\hat p_{u,u},\ldots,p_{u,t+u} )]
	\end{align}

    By definition of $t'$, $d_{u,t'}$ is the final status associated with $X_u$. Therefore, it is of interest to know the probability, noted $\mathbb P(\overline V_{u,t'})$, that $d_{u,t'}$ is equal to the oracle status.
    
    According to the law of total probabilities on the event that the assigned segment change, noted $W_{\overline u,t'}$ and defined in Eq.~\ref{eq:delay}.
    $\overline{u} = b_{t'}$
    \begin{align}
    	\mathbb P(\overline V_{u,t'}) &= 	\mathbb P(\overline V_{u,t'}|W_{\overline{u},t'})\PP (W_{\overline{u},t'}) + \mathbb P(\overline V_{u,t'}|\overline W_{\overline{u},t'})\PP (\overline W_{\overline{u},t'})\\
    	&\leq \mathbb P(\overline V_{u,t'}|W_{\overline{u},t'}) + \PP (\overline W_{\overline{u},t'})\label{eq:bound_v_prob}
    \end{align}	
    Now let's upper bound the different terms on the right-hand side of Eq.~\ref{eq:bound_v_prob}.
    According the definition of $\ell_\eta \leq m$ in in Proposition~\ref{prop:stationarity}:
    \begin{align}
    	\mathbb P(\overline V_{u,t'}|W_{\overline{u},t'}) \leq \eta \label{eq:bound_v_cond_w}
    \end{align} 
    
    Then the probability of $X_{\overline u}$ changing its assigned segment at $\llbracket t', \infty\llbracket$ is written by distinguishing the time when this change occurs at $\llbracket t', t'+2m\rrbracket$ or at $\rrbracket t'+2m, \infty \llbracket$ 
    \begin{align}
    	 \PP (\overline W_{\overline{u},t'}) &=  \PP (\exists t'' > t', \tau(t'')\cap\rrbracket\tau_{i}(t'),\tau_{i+1}(t')\llbracket \neq \emptyset  ) \\
    	 &= \PP(\exists t'+2m\geq t'' > t', \tau_{t''}\cap\rrbracket\tau_{i}(t'),\tau_{i+1}(t')\llbracket \neq \emptyset )\\
    	 &\qquad + \PP(\exists t'' > t'+2m, \tau_{t''}\cap\rrbracket\tau_{i}(t'),\tau_{i+1}(t')\llbracket \neq \emptyset )
    \end{align} 
     According to Eq.~\ref{eq:move-b} it gives:
     \begin{align}
     	\PP(\exists t'+2m\geq t'' > t', \tau_{t''}\cap\rrbracket\tau_{i}(t'),\tau_{i+1}(t')\llbracket \neq \emptyset )  \leq \eta \label{eq:bound_w_first}
     \end{align}
 
     According to the definition on $\lambda_\eta\leq m$ in Proposition~\ref{prop:stationarity} and the assumption described by Eq.~\ref{eq:all-current}.
     \begin{align}
     	 \PP(\exists t'' > t'+2m, \tau_{t''}\cap\rrbracket\tau_{i}(t'),\tau_{i+1}(t')\llbracket \neq \emptyset ) \leq \eta \label{eq:bound_w_second}
     \end{align}

    Finally, by injecting the bounding terms in Eq.~\ref{eq:bound_v_prob} using the results from Eq.~\ref{eq:bound_v_cond_w}, \ref{eq:bound_w_first} and \ref{eq:bound_w_second}, it gives:
    \begin{align}
    	\mathbb P(\overline V_{u,t'}) \leq 3\eta
    \end{align}
     
     \textbf{Proof of the second statement:}
	According to the first part of the proof, the probability that the final status is different from the oracle status is less than $3\eta$. But this property is local, valid for each $u$. The aim is to have a global property. For this, a property of ergodicity of $V_{u,t}$ is to be proved. According to \cite{blum1963strong}, it suffices to show that there exists a number $q$ such that if $|t_1-t_2|>q$ then $V_{u,t_1}$ is independent of $V_{u,t_2}$.
	
    According to Eq.~\ref{eq:breakpoint-indp}, breakpoint positions are independent beyond a distance $q_1$.
    According to the ideal BKAD operation, for a given segmentation, the statuses $d_{u,t_1}$ and $d_{u,t_2}$ are independent if $|t_1 - t_2|>m+n$.
    It follows that the random variables in the series $V_{u,t}$ are independent if $|u_1 - u_2|>m+n+q_1$.
	This completes the proof.
\end{proof}

\newpage
\section{Supplementary descriptions about the method}\label{sec:supplement-descriptions}
\subsection{Estimate the probability of segment assignment change}
\label{sec:delay-proba}
As introduced in Section~\ref{sec:th-uncertainty}, $f_\tau(\lambda)$ is the probability that the segment assignment changes when a data point is at distance $\lambda$ from the end of the time series. This probability $f_\tau(\lambda)$ is needed to build the active set containing data points with uncertain status, as described in Algorithm~\ref{alg:active-set}. In the following, a procedure is proposed to estimate $f_\tau(\lambda)$.

\label{sec:delay-proba-method}
As a reminder, the existence of $f_\tau(\lambda)$ is ensured by the stationarity assumption described in Proposition~\ref{prop:stationarity}.
However, stationarity is not sufficient to calculate these probabilities directly from historical data in the same time series and thus to estimate $\hat f_\tau$. It must also be assumed that the series $\mathbb{1}\left[\W_{t-\lambda,t} \right]$ is ergodic. 
\begin{proposition}[Ergodicity]\label{prop:ergodicity}
	Assume $\mathbb{1}\left[\W_{t-\lambda,t} \right]$ is stationary and ergodic. Then
	\begin{align}\label{eq:erg-seg-questioning}
		\mathbb P[\W_{t-\lambda,t}] = \lim_{T \rightarrow \infty} \frac{1}{T}\sum_{\tilde t=1}^{T}\mathbb{1}[\W_{\tilde t-\lambda, \tilde t}] 
	\end{align}
\end{proposition}
The conclusion of the Proposition~\ref{prop:ergodicity} follows directly from the definition of ergodic process \cite{gray2009probability}.

This ergodicity property can be used to derive a learning algorithm. The time series is split into two parts: historical and recent data. The historical data set is built using the first $\tilde T$ data points of the time series.
The estimation of $f_\tau$ is based on the previous segment assignment changes that occurred while detecting breakpoints on historical data.
To estimate this probability, the list of all previous segmentations $\mathcal{D} =\left(\hat\tau_{1},\ldots, \hat\tau_{\tilde T}\right)$ is used. 
Assuming stationarity and ergodicity of $\mathbb{1}\left[\W_{t-\lambda,t} \right]$, where the event $\W_{t-\lambda,t}$ is described in Eq.~\ref{eq:delay}, these historical data are used to estimate $f_\tau(\lambda)$ using Eq.~\ref{eq:erg-seg-questioning}. 
\begin{align}
	\PP\left[\W_{t-\lambda,t} \right] \approx  \frac{1}{\tilde T}\sum_{t=1}^{\tilde T} \mathbb{1}\left[\W_{\tilde t, \tilde t-\lambda}^{\tilde T}\right] = \hat f_\tau(\lambda)
\end{align} 
where $\W_{\tilde t, \tilde t-\lambda}^{\tilde T} = \left\lbrace\exists t' \in \llbracket \tilde t,\tilde T\rrbracket,\quad \hat \tau_{t'} \cap \llbracket \hat b_{\tilde t}, \tilde t-\lambda\rrbracket \neq \emptyset\right\rbrace$.

However, to improve computation time, the following expression of $\W_{\tilde t, \tilde t-\lambda}^{\tilde T}$ is preferred:
\begin{equation}
	\W_{\tilde t-\lambda,\tilde t}^{\tilde T}= \left\lbrace\left(\bigcup_{\tilde T\geq t'>\tilde t}\hat \tau_{t'}\right) \cap \llbracket \hat b_{\tilde t}, \tilde t-\lambda\rrbracket \neq \emptyset\right\rbrace\label{eq:exact-delay}
\end{equation}
With this formulation, each breakpoint is checked only once to see if it belongs to $\llbracket \hat b_{\tilde t}, \tilde t-\lambda\rrbracket$. Indeed, many breakpoints remain at the same position from one step to the next step while applying the breakpoint detection procedure.

Algorithm~\ref{alg:exact-delay} implements Eq.~\ref{eq:exact-delay} to give an estimation of $f_\tau$. Where $I_{\tilde t,u} = \mathbb{1}[\W_{\tilde t, u}^{\tilde T}]$ and $S_{\lambda} = \sum_{\tilde t=\lambda}^{\tilde T} I_{\tilde t, t-\lambda}$  so $\hat f_\tau(\lambda) = \frac{S_\lambda}{\tilde T-\lambda}$.

\begin{algorithm}[tb]
	\begin{algorithmic}
		\Require $(\tau(\tilde t))_1^{\tilde T}$ list of successive segmentations
		\State $I, S \gets 0$
		\State $\tau_{global} \gets \emptyset $
		\For{$\tilde t\in[\tilde T,1]$}
		\State $\tau_{global} \gets \tau_{global}\cup \hat\tau({\tilde t})$
		\For{$u\in[\hat b_{\tilde t}, \tilde t]$}
		\For{$b' \in \tau_{global} $ }  \Comment{$b'$ is a breakpoint} 
		\If{$ \hat b_{\tilde t}<b' \leq u$}
		\State $I_{\tilde t,u} \gets 1$ 
		\EndIf 
		\EndFor
		\EndFor
		\EndFor
		\For{$\lambda \in [0,\tilde T]$}
		\For{$\tilde t \in [\lambda, \tilde T]$}
		\State{$S_{\lambda}\gets S_\lambda + I_{\tilde t,\tilde t-\lambda}$}
		\EndFor
		\State $\hat f_{\tau,\lambda} \gets S_{\lambda}/(\tilde T-\lambda)$
		\EndFor
		\\{\bf Output:} $\hat f_{\tau,\lambda}$ list of $\hat f_\tau(\lambda)$ values for different $\lambda$
		\\\Return $\hat f_{\tau,\lambda}$
	\end{algorithmic}
	\caption{Exact computation of probability of segment assignment change.}\label{alg:exact-delay}
\end{algorithm}

The complexity of the exact computation of $\hat f_\tau$, described in Algorithm~\ref{alg:exact-delay}, is quadratic in time and space, which is a drawback regarding any practical use.
Another version of the implementation of $\hat f_\tau$  is given in Algorithm~\ref{alg:efficient-rt}  which is more convenient for an online context since it is linear in time and space.

Indeed, by observing the evolution of the breakpoints over time (not shown in this paper), it appears that the position of the last breakpoint is the most likely to change, while that of the other breakpoints are generally stable.
This leads us to modify the characterization of the ``assigned segment change'' event by considering only the change caused by the last breakpoint instead of the entire segmentation. 
\begin{equation}\label{assum.lastbreal}
	\forall t ,\lambda \in \llbracket 1,T\rrbracket^2\qquad \mathbb{1}[\W_{t,t-\lambda}^T]= \mathbb{1}\left[\exists t' \in \llbracket t,T\rrbracket,\quad  \hat b_t < \hat b_{t'}\leq t-\lambda \right]\tag{{\bf Last}}
\end{equation}

Under this assumption, the computation of $\hat f_{\tau}(\lambda)$ can be simplified using Proposition~\ref{prop:ftau-simplicafication}.
\begin{proposition}\label{prop:ftau-simplicafication}
	Let $(X_t)_{1\leq t \leq T}$ be a time series of length $T$. Let $(\hat\tau(t))_{1\leq \tilde t\leq \tilde T}$ be the sequence of successive segmentations of the time series. Let $(\mathbb{1}[W_{t,u}^T])_{1 \leq t \leq T, 1 \leq u \leq T}$ the family of ``assigned segment change'' events.
	Assume that the assumption \eqref{assum.lastbreal} is verified.
	Then the estimator $\hat f_\tau$ described in Eq.~\ref{eq:exact-delay} is computed as
	\begin{align}
		\hat f_\tau(\lambda) = \frac{1}{\tilde T}\sum_{\tilde t=1}^{\tilde T}\mathbb{1}\left[r_{\tilde t}>\lambda \right]
	\end{align}
	where $r_{\tilde t}$ is the maximum distance from the the end of the time series with $X_t$ having segment reassigned. It is computed as:
	\begin{align}
		r_{\tilde t} = \max_{t'>\tilde t, \hat b_{t'}> \hat b_{\tilde t}, \hat b_{t'}<\tilde t}\tilde t-\hat b_{t'}
	\end{align}
\end{proposition}
Notice that $r_{\tilde t}$ does not depend on $\lambda$. It is sufficient to calculate $r_{\tilde t}$ once for all $\lambda$. Therefore, it's easy to deduce the value of $f_\tau(\lambda)$ for all $\lambda$. The most demanding part is the computation of $r_{\tilde t}$. Two implementations of $r_{\tilde t}$ computation are proposed. Algorithm~\ref{alg:naive-rt} gives the most naive version, each $r_{\tilde t}$ is calculated one after the other. The problem is that the calculation of $r_{\tilde t}$ is itself linear in the length of the series. Therefore, the time complexity is quadratic.
Algorithm~\ref{alg:efficient-rt} improves the computation by swapping the two loops. This limits the total number of comparisons performed. In the second loop, $\tilde t$ takes on the values between $b_t'$ and $t'$. The number of values taken by $\tilde t$ is the length of a segment, not a the length of the time series. Algorithmic complexity is therefore linear. 

\begin{proof}[Proof of Proposition \ref{prop:ftau-simplicafication}]
	Based on Eq.~\ref{eq:exact-delay},  the estimator $\hat f_{\tau}$ is given by: 
	\begin{align*}
		\hat f_\tau(\lambda) = \frac{1}{\tilde T}\sum_{t=1}^{\tilde T}\mathbb{1}\left[\W_{\tilde t,\tilde t-\lambda}^{\tilde T}\right]
	\end{align*}
	With assumption \eqref{assum.lastbreal} it gives:
	\begin{align*}
		\mathbb{1}\left[\W_{\tilde t,\tilde t-\lambda}^{\tilde T}\right]&= \mathbb{1}\left[\exists t' \in \llbracket \tilde t,\tilde T\rrbracket,\quad  \hat b_{\tilde t} < \hat b_{t'}\leq \tilde t-\lambda \right]
	\end{align*}
	The inequality $\hat b_{\tilde t} < \hat b_{t'} \leq \tilde t-\lambda$ is equivalent to $\hat b_{\tilde t} < \hat b_{t'}$ and $\tilde t-\hat b_{t'}>\lambda$ which gives
	\begin{align*}
		\mathbb{1}\left[\W_{\tilde t,\tilde t-\lambda}^{\tilde T}\right]&= \mathbb{1}\left[\bigcup_{t'>\tilde t, \hat b_{t'}> \hat b_{\tilde t}, \hat b_{t'}<\tilde t}\tilde t-\hat b_{t'}>\lambda \right]\\
	\end{align*}
	Since, a set contains a number greater than $\lambda$, if and only if its maximum is greater than $\lambda$, it gives:
	\begin{align*}
		\mathbb{1}\left[\W_{\tilde t,\tilde t-\lambda}^{\tilde T}\right] &= \mathbb{1}\left[\left(\max_{t'>\tilde t, \hat b_{t'}> \hat b_{\tilde t}}\tilde t-\hat b_{t'}\right)>\lambda \right]
	\end{align*}
	Since $\lambda>0$, when $\hat b_{\tilde t} < \hat b_{t'}$ and $\tilde t-\hat b_{t'}>\lambda$  it also implies that $\tilde t\geq\hat b_{t'}$ so
	$\mathbb{1}\left[\left(\max_{t'>\tilde t, \hat b_{t'}> \hat b_{\tilde t}}\tilde t-\hat b_{t'}\right)>\lambda \right] = \mathbb{1}\left[\left(\max_{t'>\tilde t, \hat b_{t'}> \hat b_{\tilde t}, \hat b_{t'}<\tilde t}\tilde t-\hat b_{t'}\right)>\lambda \right]$. 
	The number $r_{\tilde t}$ is introduced as equal to $\max_{t'>\tilde t, \hat b_{t'}> \hat b_{\tilde t}, \hat b_{t'}<\tilde t}\tilde t-\hat b_{t'}$. 
	The $\hat f_\tau$ estimator can be written as follows
	\begin{align*}
		\hat f_\tau(\lambda) = \frac{1}{\tilde T}\sum_{\tilde t=1}^{\tilde T}\mathbb{1}\left[r_{\tilde t}>\lambda \right]
	\end{align*}
\end{proof}

\begin{minipage}{0.45\linewidth}
	\begin{algorithm}[H]
		\begin{algorithmic}
			\For{$\tilde t$ in $\llbracket 1, \tilde T\rrbracket$}
			\For{$t'$ in $\llbracket \tilde t, \tilde T\rrbracket$}
			\If{$\hat b_{t'} < \tilde t$ and $\hat b_{\tilde t} > \hat b_{t'}$}
			\State{$r_{\tilde t} = \max(r_{\tilde t}, \tilde t-\hat b_{t'})$}
			\EndIf
			\EndFor
			\EndFor
		\end{algorithmic}
		\caption{Naive computation of $r_{\tilde t}$.}
		\label{alg:naive-rt}
	\end{algorithm}
\end{minipage}
\begin{minipage}{0.49\linewidth}
	\begin{algorithm}[H]
		\begin{algorithmic}
			\For{$t'$ in $\llbracket 1, \tilde T\rrbracket$}
			\For{$\tilde t$ in $\llbracket \hat b_{t'}, t'\rrbracket$}
			\If{$\hat b_{\tilde t} > \hat b_{t'}$}
			\State{$r_{\tilde t} = \max(r_{\tilde t}, \tilde t-\hat b_{t'})$}
			\EndIf
			\EndFor
			\EndFor
		\end{algorithmic}
		\caption{Efficient computation of $r_{\tilde t}$.}
		\label{alg:efficient-rt}
	\end{algorithm}
\end{minipage}

\begin{figure}[!b]
	\centering
	\includegraphics[width=0.6\linewidth]{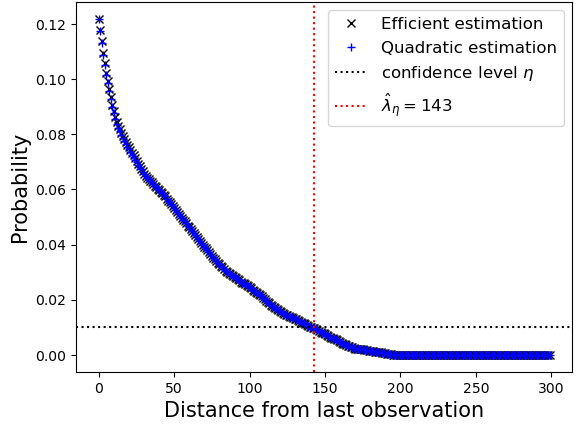}
	\caption{Probability of assignment change as a function of distance to time series end.}
	\label{fig:activ-exp-req}
\end{figure}

Figure~\ref{fig:activ-exp-req} shows an example of the estimated probability of a segment assignment change according to the two  estimation Algorithms~\ref{alg:exact-delay} and \ref{alg:efficient-rt}.
The two algorithms give almost identical results, as shown in Figure~\ref{fig:activ-exp-req}. The selected $\hat\lambda_\eta$ is equal to $143$, in both cases. This supports the assumption that \eqref{assum.lastbreal} is verified. 
This result is used in the rest of this article to define the value of $\hat \lambda_\eta$, the minimum length of the current segment in Algorithm~\ref{alg:active-set}.
More details about the experiment can be found in the Section 1 in the Supplementary Material.

\subsection{Estimate the probability that the point will have a status different from that of the oracle if the point is assigned to the same segment as the oracle.}\label{sec:stable-seg-change}
As introduced in Section~\ref{sec:th-uncertainty}, $f_d(\ell)$ is the probability that the status of a point changes under the conditions the last breakpoint remains unchanged and the segment cardinality is equal to $\ell$. This probability $f_d(\ell)$ is necessary to build the active set containing data points with uncertain status, as described in Algorithm~\ref{alg:active-set}. In this section, a procedure to estimate $f_d(\ell)$ is proposed.

\begin{figure}[tb]
	\centering
	\begin{subfigure}{0.35\linewidth}
		\centering
		\includegraphics[width=0.7\linewidth]{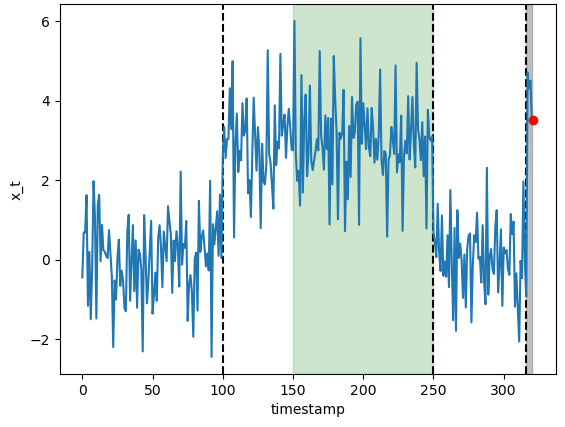}
		\caption{Starting new segment with small length.}
		\label{fig:zscore-segsize-a}
	\end{subfigure}
	\begin{subfigure}{0.35\linewidth}
		\centering
		\includegraphics[width=0.7\linewidth]{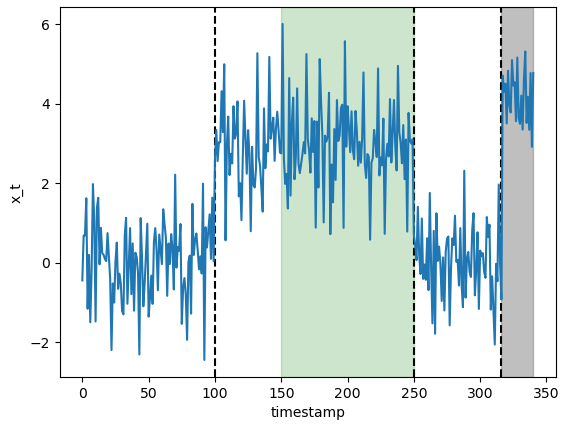}
		\caption{Data points added to the current segment becoming larger.}
		\label{fig:zscore-segsize-b}
	\end{subfigure}
	\begin{subfigure}{0.35\linewidth}
		\centering
		\includegraphics[width=0.7\linewidth]{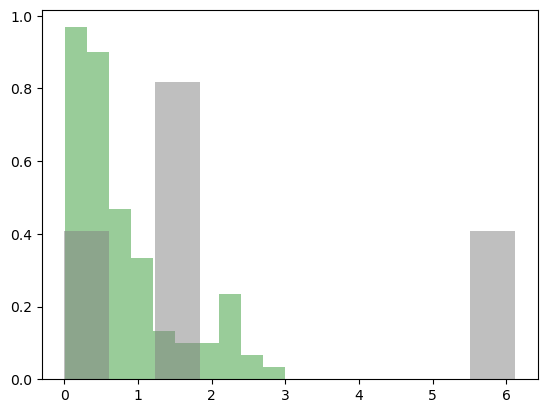}
		\caption{Z-score for small segment}
		\label{fig:zscore-segsize-c}
	\end{subfigure}
	\begin{subfigure}{0.35\linewidth}
		\centering
		\includegraphics[width=0.7\linewidth]{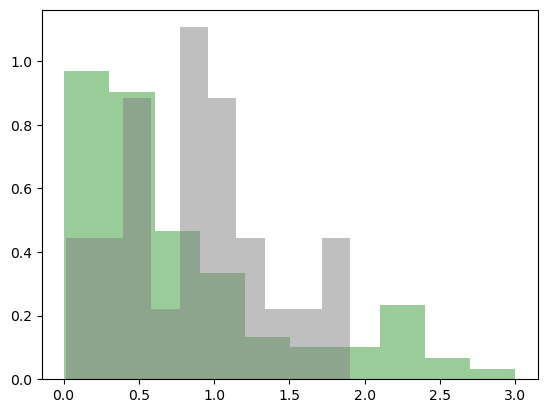}
		\caption{Z-score for larger segment}
		\label{fig:zscore-segsize-d}
	\end{subfigure}
	\caption{Atypicity score estimation according to the length of the current segment.}
	\label{fig:zscore-segsize}
\end{figure}

Figure~\ref{fig:zscore-segsize}  illustrates how the the length of the current segment  has an influence on the accuracy of the atypicity score estimation and consequently on the uncertainty of a data point status.
Indeed, Figure~\ref{fig:zscore-segsize-a} shows a time series with a newly detected current segment highlighted in gray color, and a calibration set in green color inside the previous segment. The atypicity scores, $z$-scores based on the mean and the standard deviation, are shown in Figure~\ref{fig:zscore-segsize-c} computed for the current segment in gray and the calibration set in green.
Since the current segment has few points, its $z$-score estimation shows a high discrepancy compared to the score distributions of the calibration set, despite the fact that there are no anomalies. 
As shown in Figure~\ref{fig:zscore-segsize-d}, when new data points are added, the estimation of the abnormality score of the current segment is more accurate.

This example highlights that the status of a data point can change even if the breakpoints remain unchanged, whereas Appendix~\ref{sec:delay-proba} deals only with the case where the change in status is due to a change in the detected breakpoints. Uncertainty also comes from having too few points in the current segment, leading to score estimation errors. Estimating the probability $\hat f_d(\ell)$ is useful to build the active set that takes this into account. 
The following section suggests a procedure to estimate the probability $f_d(\ell)$.

\label{sec:f2method}
The method is based on the learning phase using a set $\mathcal D$ of historical detected segments having a low probability to change (using final step $T$). This training set $\mathcal \mathcal D$ is defined by,
\begin{align}
	\mathcal{D} = \lbrace (X_1,\ldots, X_{\hat \tau_{1}(T)}), (X_{\hat \tau_{1}(T)+1},\ldots, X_{\hat \tau_{2}(T)}), \ldots, (X_{\hat \tau_{D-1}(T)+1},\ldots,X_{\hat \tau_{D}(T)})\rbrace
\end{align}
In the following, the training procedure  is based on six different steps needed to estimate the $\hat f_d(\ell)$ probability.
Let $\overline a$ be the NCM (Non Conformity measure), used to define the atypicity score, as described in Appendix~\ref{sec:score}. As a reminder, $\overline a(S,x)$, measures the ``nonconformity'' between the set $S$ and the point $x$.

The principle of the training phase is to simulate, using resampling, numerous examples where the current segment changes from a length $\ell$ to the final length. At each case, anomaly detection is applied to the test set of cardinality $m$ and the proportion of statuses that have changed by modifying the length of the current segment is measured. The status obtained from the maximum size segment is the oracle status. By comparing it with the status obtained with the $\ell$ size segment, the confidence score can be approximated.
Since the breakpoints are assumed to be stable, the simulation is inspired by the description of the detector given in Section~\ref{alg:generalise-sbad}, without the parts concerning breakpoint detection.
These steps are repeated $B$ times. Let $b \in \llbracket 1, B\rrbracket$:
\begin{enumerate}[label=Step~\arabic*:]
	\item Figure~\ref{fig:resampling-step-by-step-a} illustrates that two segments are resampled from the historical data set $\mathcal D$. $\mathcal S_{1,b}$ is considered as the calibration set and $\tilde{\mathcal S}_{2,b}$ as a current segment if the whole time series where observed (see Definition~\ref{def:oracle-status}) in the simulation:
	$$\mathcal S_{1,b},\tilde{\mathcal S}_{2,b}  \sim U(\mathcal D) $$
	\item The current segment $\tilde{\mathcal S}_{2,b}$ is sub-sampled into a smaller segment of length $\ell$ and noted $ S_{2,\ell,b}$, as shown in Figure~\ref{fig:resampling-step-by-step-c}. $ S_{2,\ell,b}$ is considered as the same segment than $\tilde{\mathcal S}_{2,b}$ without knowledge of the whole time series, having only $\ell$ points, . 
	$$\mathcal S_{2,\ell,b}\sim U(\tilde{\mathcal S}_{2,b})$$.
	\item The current segment $\tilde{\mathcal S}_{2,b}$ is sub-sampled into an other segments of length $m$ and noted $\overline S_{2,m, b}$. $\overline S_{2,m,b}$ is considered as the test set.  
	$$\overline {\mathcal S}_{2,m,b}\sim U(\tilde{\mathcal S}_{2,b})$$
	\item As illustrates in Figure~\ref{fig:resampling-step-by-step-d} and \ref{fig:resampling-step-by-step-e}, the scores of the three segments are computed:
	\begin{itemize}
		\item The score of $\mathcal S_{1,b}$ (``calibration set''): 
		\begin{align*}
			\forall i \in \llbracket 1, n\rrbracket, X_i \in  \mathcal S_{1,b},\quad c_{i,b} = \overline a(Y_i, \mathcal S_{1,b}\textbackslash \{X_i\})
		\end{align*}
		\item The score of the test set using $\tilde{\mathcal S}_{2,b}$ as training set: 
		\begin{align*}
			\forall i \in \llbracket 1, m\rrbracket, Y_i \in \overline S_{2,m,b},\quad \tilde s_{i,b} = \overline a(Y_i, \tilde{\mathcal S}_{2,b}\textbackslash \{Y_i\})
		\end{align*}
		\item The score of the test set using $S_{2,\ell, b}$ as a training set: 
		\begin{align*}
			\forall i \in \llbracket 1  m\rrbracket, Y_i \in \overline S_{2,m,b},\quad  s_{i,\ell,b} = \overline a(Y_i, S_{2,\ell,b}\textbackslash \{Y_i\})
		\end{align*}
	\end{itemize}
	\item Figure~\ref{fig:resampling-step-by-step-f} illustrates that the empirical $p$-values are computed for the two scores obtained from the test set  using the same calibration set:
	\begin{itemize}
		\item $p$-values of the test set when using the complete current segment as training set: $\forall i \in \llbracket 1, m\rrbracket, \tilde p_{i,b} = \frac{1}{n}\sum_{j=1}^n\mathbb{1}[\tilde s_{i,b}<c_{j,b}]$
		\item $p$-values of the test set when using the length $\ell$ sub-sample of the current segment as training set:  $\forall i \in \llbracket 1, m\rrbracket,  p_{i,\ell, b} = \frac{1}{n}\sum_{j=1}^n\mathbb{1}[ s_{i,\ell,b}<c_{j,b}]$
	\end{itemize}
	\item Detect the anomalies in the two cases, by applying the Benjamini-Hochberg procedure $\hat \varepsilon_{BH_\alpha}$ on the estimated $p$-values, as shown in Figure~\ref{fig:resampling-step-by-step-g}:
	\begin{itemize}
		\item In case the training set is the entire current segment:
		\begin{align*}
			\forall i \in \llbracket 1,m\rrbracket,\quad \tilde d_{i,b}=\mathbb{1}[\tilde p_{i,b}<\hat\varepsilon_{BH_\alpha}(\tilde p_{1,b},...,\tilde p_{m,b})]
		\end{align*}
		\item In case the training set is the sub-sample of cardinality $\ell$: 
		\begin{align*}
			\forall i \in \llbracket 1, m\rrbracket, \quad  d_{i,\ell, b}=\mathbb{1}[ p_{i,\ell,b}<\hat \varepsilon_{BH_\alpha}(p_{1,\ell,b},...,p_{m, \ell,b})]
		\end{align*}
	\end{itemize}
	\item The number of decisions that differ between the two cases, $\tilde{\mathcal S}_{2,b}$  or $\mathcal S_{2,\ell,b}$ used as the training set, is computed:
	$$n_d = \sum_{i=1}^m \mathbb{1}[\tilde d_{i,b} \neq  d_{i,\ell,b}] $$ 
\end{enumerate}
The training procedure simulates the behavior of the online anomaly detector:
$\mathcal S_1$ plays the role of the calibration set. $\tilde{\mathcal S}_2$ plays the role of current segment with knowledge of the whole time series. $\mathcal S_{2,\ell}$ plays the role of the current segment at the beginning of a new segment, that contains only $\ell$ points. The first $m$ elements $Y_1, \ldots, Y_m$ from $\tilde{\mathcal S}_2$ constitute the test set.

\newpage
\begin{figure}[H]
	\centering
	\begin{subfigure}{0.8\linewidth}
		\includegraphics[width=0.9\linewidth]{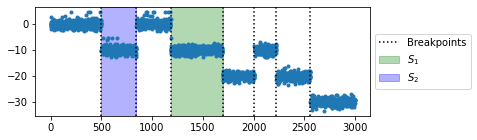}
		\caption{Step 1: Segments resampling}
		\label{fig:resampling-step-by-step-a}
	\end{subfigure}
	\begin{subfigure}{0.8\linewidth}
		\includegraphics[width=0.9\linewidth]{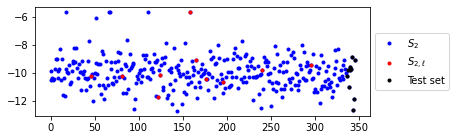}
		\caption{Step 2 and 3: Sub-sampling}
		\label{fig:resampling-step-by-step-c}
	\end{subfigure}
	\begin{subfigure}{0.8\linewidth}
		\includegraphics[width=0.9\linewidth]{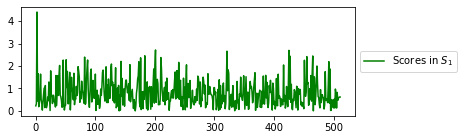}
		\caption{Step 4: Calibration set scoring}
		\label{fig:resampling-step-by-step-d}
	\end{subfigure}
	\begin{subfigure}{0.8\linewidth}
		\includegraphics[width=0.9\linewidth]{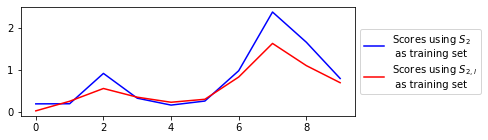}
		\caption{Step 5: Test set scoring}
		\label{fig:resampling-step-by-step-e}
	\end{subfigure}
	\begin{subfigure}{0.8\linewidth}
		\includegraphics[width=0.9\linewidth]{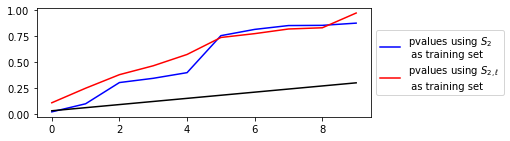}
		\caption{Step 6: $p$-value estimation}
		\label{fig:resampling-step-by-step-f}
	\end{subfigure}
	\begin{subfigure}{0.8\linewidth}
		\includegraphics[width=0.9\linewidth]{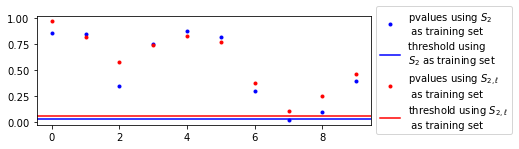}
		\caption{Step 7: Anomaly detection}
		\label{fig:resampling-step-by-step-g}
	\end{subfigure}
	\caption{Illustration of the different steps of the training procedure to estimate the status change  probability under stable breakpoints.}
	\label{fig:resampling-step-by-step2}
\end{figure}

Assuming the score are iid as stated in Definition~\ref{def:score-homogeneity}, by repeating this resampling process many times, as the length of the time series converges to infinity, the proportion of status changes converges to the expectation, according to the law of large numbers \cite{csorgő1992law}:
\begin{align*}
	\lim_{T,B \rightarrow \infty}\frac{1}{mB}\sum_{b=1}^B \sum_{j=1}^m \mathbb{1}[\tilde d_{j,b} \neq d_{j,\ell,b}] = \mathbb E_{\mathcal S_1,\mathcal S_2\sim U(\mathcal D)} \mathbb E_{\mathcal S_{2,\ell} \sim U(\mathcal S_2)} \sum_{i=1}^m \mathbb{1}[d_i \neq \tilde d_i] 
\end{align*}

Under the assumptions of score stationarity stated in Definition~\ref{def:score-homogeneity}, the limit is equal to $f_{d}(\ell)$. Indeed, under score stationarity, the calibration set can be built from any segment of the time series. This implies that it is possible to use described training procedure as an estimator of $\hat f_{d}(\ell)$.
\begin{align*}
	\hat f_d(\ell) = \frac{1}{mB}\sum_{b=1}^B \sum_{j=1}^m \mathbb{1}[\tilde d_{j,b} \neq d_{j,\ell,b}] \approx f_d(\ell) 
\end{align*}

Figure~\ref{art2-activ-exp-conf-zscore} illustrates results when applying this algorithm on Gaussian data. Three line charts representing the probability of status change as a function of the current segment length in relation to the initial status: (a) the status is normal, (b) the status is abnormal and (c) unknown status. In the unknown status, Figure~\ref{art2-fig:activ-exp-conf-zscore-general} shows clearly that the probability of status change decreases with the length of the current segment. This probability is higher when the status is abnormal, as shown in Figure~\ref{art2-fig:activ-exp-conf-zscore-rej}. Nevertheless,  with a segment length of $100$, the probability is less than $1\%$. This result is used in the rest of this article to define the value of $\hat \ell_\eta$, the minimum length of the current segment in Algorithm~\ref{alg:active-set}. More details about this experiment can be found in the Supplementary Material in Section 2.

\begin{figure}[tb]
	\centering
	\begin{subfigure}{0.3\textwidth}
		\centering
		\includegraphics[width=0.9\linewidth]{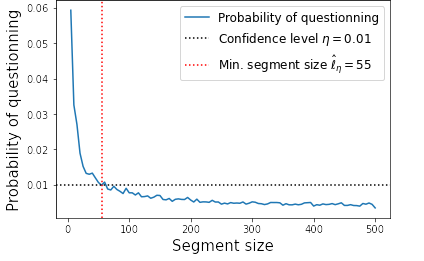}
		\caption{normal status}
		\label{art2-fig:activ-exp-conf-zscore-nonrej}
	\end{subfigure}%
	\begin{subfigure}{0.3\textwidth}
		\centering
		\includegraphics[width=0.9\linewidth]{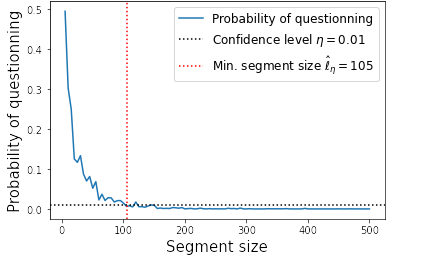}
		\caption{abnormal status}
		\label{art2-fig:activ-exp-conf-zscore-rej}
	\end{subfigure}%
	\begin{subfigure}{0.3\textwidth}
		\centering
		\includegraphics[width=0.9\linewidth]{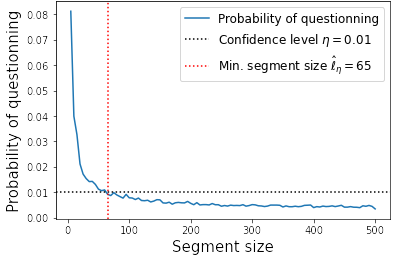}
		\caption{unknown status}
		\label{art2-fig:activ-exp-conf-zscore-general}
	\end{subfigure}%
	\caption{Probability that status changes under stable breakpoints as a function of segment length, for Gaussian data.}
	\label{art2-activ-exp-conf-zscore}
\end{figure}

\subsection{Breakpoint estimation}
\label{sec:breakpoint}
As described at Section~\ref{sec:modelisation-problem} the time series $(X_t)_{1 \leq t \leq T}$ has $D$ breakpoints denoted $\tau_1, ..., \tau_{D+1}$. The segment $X_{\tau_i}^{\tau_{i+1}-1}$ is said homogeneous. Informative features for anomaly detection, such as the mean or the variance, can be extracted if the breakpoints are correctly identified.  
A good breakpoint detector is important to increase the accuracy of anomaly detection.
If a shift is not well detected, the analyzed segment will be heterogeneous and the estimation of the law under $\mathcal H_0$ will be biased. 
If too many breakpoints are detected in a segment while it is homogeneous, the analyzed segments will contain fewer points and the variance of the predictions will be too high.
To maximize the performance of the anomaly detection, the number and the locations of breakpoints have to be accurately estimated. 
The article \cite{truong2020selective} is a review of existing offline breakpoint detectors. The authors show that a breakpoint detector can be described as an optimization problem, using three notions.
\begin{itemize}
	\item \textbf{Cost function}. A cost function $\mathcal C(\cdot)$ measures the homogeneity of a given subseries $X_{t_1}^{t_2}$. With a well chosen cost function, $\mathcal C(X_{t_1}^{t_2})$ is high when there is at least one breakpoint between $t_1$ and $t_2$. The cost function is low when there is no breakpoint in this subseries. 
	\item \textbf{Search method}: The search method enables to explore a set of possible segmentations, denoted $\mathcal T$, of the optimization problem. Each search method is a trade-off between accuracy and computational complexity \cite{truong2020selective}.
	\item \textbf{Penalty function}: The penalty function is useful when the number of breakpoints is unknown. It avoids overestimating the number of breakpoints by penalizing segmentations with a large number of breakpoints. The penalty function $pen(\cdot)$ increases based on the number of breakpoints, noted $D_\tau$.
\end{itemize}
The segmentation returned by the breakpoint detector is the one that minimizes the penalized cost function among the explored solutions:
\begin{align}\label{eq:bkpts-criterion}
	\hat \tau \in \arg\min_{\tau \in \mathcal{T}} \sum_{i=1}^{D_\tau}\mathcal C(X_{\tau_i}^{\tau_{i+1}-1}) + pen(D_\tau)
\end{align}

In this article, the Kernel Change Point (KCP) introduced in \cite{arlotKernelMultipleChangepoint} is used for its advantages. 
The kernel-based cost function could be used for any kind of time series, univariate or multivariate, without changing the breakpoint detector. To handle the diversity of time series data, the kernel and its hyperparameters have to be carefully chosen to be able to detect any kind of change points in the time series.
This accuracy is guaranteed by the oracle inequality given in \cite{arlotKernelMultipleChangepoint}.
For a given segmentation $\tau$ and a kernel $K$, the cost is given by:
\begin{equation}
	\hat R(\tau) = \frac{1}{t}\sum_{u=1}^t K(X_u,X_u) - \frac{1}{t}\sum_{i=1}^{D_\tau} \frac{1}{\tau_{i+1} - \tau_{i}}\sum_{u,v=\tau_i}^{\tau_{i+1}-1}K(X_u, X_v)
\end{equation}
First, the candidate segmentations that minimize the criterion are identified for each possible number of $D$ segments. $\mathcal T^D$ is the space of all candidate segmentations with $D$ segments, $\hat \tau_{D,t}$ is the best candidate segmentation with $D$ segments and $L_{D,t}$ is the cost associated with this segmentation.
\begin{align*}
	L_{D,t} &= \min_{\tau \in \mathcal T^D} \hat R(\tau)\\
	\hat \tau_{D,t} &= \arg\min_{\tau \in \mathcal T^D} \hat R(\tau)\\
\end{align*}
To estimate the number of segments and thus the best segmentation, one searches for the segmentation $\hat \tau_{D,t}$ that minimizes the penalized criterion described in Eq.~\ref{eq:bkpts-criterion}. The penalty function is given by:
\begin{equation}
	pen(\tau) = r_1 D_\tau + r_2 \log \begin{pmatrix}
		t-1\\D_\tau - 1
	\end{pmatrix}
\end{equation}
where the coefficients $r_1$ and $r_2$ are estimated by fitting the penalty function on the estimated cost for over-segmented segmentations \cite{baudry2012slope}.

KCP is designed as an offline breakpoint estimator. By using Dynamic Programming, the segmentation costs can be estimated without performing the same computation between time $t$ and $t+1$ as described in \cite{celisse2018new}. This feature is necessary to be applied in an online anomaly detection. The data driven penalty function enables good accuracy in estimating the number of breakpoints.
The breakpoints are detected by solving the optimization problem with the algorithm:
\begin{algorithm}[tb]
	\begin{algorithmic}
		\Require $T>0$, $(X_t)$ time series, $\mathcal C$ Kernel based cost function, $D_{max}$ maximum breakpoint number explored and $SlopeHeuristic$ implement the slope heuristic.
		\For{$t \in \llbracket 1, T \rrbracket$}
		\For{$D \in \llbracket 1, D_{max} \rrbracket$}
		\State{$L_{D, t} \gets \min_{t' \leq t} L_{D-1, t'} + \mathcal C_{t',t}$}
		\State{$\hat \tau_{D, t} \gets \arg\min_{t' \leq t} L_{D-1, t'} + \mathcal C_{t',t}$}
		\EndFor
		\State{$c_1,c_2 \gets SlopeHeuristic(L)$}
		\State{$\hat D \in \arg\min_{D} L_{D,t} + c_1 D + c_2 \log \begin{pmatrix}
				t-1\\D - 1\end{pmatrix}$}
		\State{$\hat \tau_{t} \gets \hat \tau_{\hat D, t}$}
		\EndFor
		\\{\bf Output:} $\forall t \in \llbracket 1, T\rrbracket, (\tau(t))$ estimated segmentation at each time step.
	\end{algorithmic}
	\caption{Dynamic Programming for breakpoint detection.}
\end{algorithm}

The main degree of freedom in KCP is the choice of the kernel. Characteristic kernels \cite{fukumizu2009kernel}, like the Gaussian kernel, are able to detect any kind of change: shift in the mean, shift in the variance, shift in the third moment,\dots
\begin{equation}
	K(x,y) = \exp(- \frac{||x-y||^2}{2 h^2})
\end{equation}
However, due to the fact that the number of points within a segment is finite, the performance of a characteristic kernel depends on the choice of hyperparameters. In the case of the Gaussian kernel, the only parameter is the bandwidth $h$. For changes that occur in the mean, the \emph{median heuristic}, shown in Eq.~\ref{eq:median-heuristic}, gives good results \cite{garreau2017large}. Defining a method to select the most relevant kernel for any kind of breakpoint is still an open question.
\begin{equation}\label{eq:median-heuristic}
	h = \mbox{median}_{i \neq j}(||X_i-X_j||)
\end{equation}

\subsection{Atypicity score}
\label{sec:score}
In this paper, an anomaly is a data point that does not follow the reference distribution of the segment to which it belongs. To construct an atypicity score that is higher for abnormal points, a point must be compared to the rest of the segment.The Nonconformity Measure (NCM) from \cite{shafer2008tutorial} is introduced. The Nonconformity Measure $\overline a$, is a real valued function $\overline a(z, B)$ that measure how different $z$ is from the set $B$. A nonconformity measure can be used to compare a data point with the rest of the segment. If all points within a segment are generated by the reference distribution, then the Nonconformity Measure provides an atypicity score for this segment.
\begin{align}\label{eq:ncm}
	\forall i \in \llbracket 1 ,D\rrbracket,\quad \forall \llbracket \tau_i, \tau_{i+1}\llbracket, \quad a(X_t) = \overline a(X_t, \lbrace X_{\tau_i}, \ldots, X_{\tau_{i+1}-1}\rbrace \backslash \lbrace X_t \rbrace) 
\end{align}

The following properties are required to enable the usage of the nonconformity measure to build a good atypicity score:
\begin{itemize}
	\item anomalies should have higher atypicity score than normal data points. 
	\item the NCM must be robust \cite{staudte2011robust, rousseeuw2018anomaly, zhou2017anomaly} to the presence of anomalies. The anomalies present in the segment do not affect the value of the returned measure. 
	\item  the values returned between different segments must be comparable, so that a $p$-value can be estimated, with a calibration set containing values from different segments. The iid property of scoring is introduced in Definition~\ref{def:score-homogeneity}, to formalize this idea.
\end{itemize}
The property of score stationarity depends on the time series. For example, the $z$-score with true known mean and standard deviation satisfies the stationarity property only if the changes generated by the  breakpoints are shifts in the mean or in the standard deviation. If the change occurs in higher moments, the property is not satisfied.
Furthermore, the property is not satisfied for the $z$-score if the mean and standard deviation need to be estimated.
Since the stationarity of the score is difficult to obtain, it is approached with the following strategies:
\begin{itemize}
	\item Ensure that the segment contains enough points to ensure the convergence of the nonconformity measures. For example, since the mean and variance must be estimated, the $z$-score is not stationary. However, if these parameters converge to the true mean and standard deviation, then the $z$-score can be considered stationary once the segments have enough points. The faster convergence is achieved, the easier it is to ensure the stationarity property. An NCM is said to be efficient when convergence is achieved for a low number of points.
	\item Instead of using the entire segment, the training set can be built by resampling a specified number of points. It can be used on NCM that  are highly dependent on the training set cardinality, like $k$NN.
	\item  Rather than trying to ensure that the score distribution is identical in each segment, identify the segments with the most similar distribution, as described in Appendix~\ref{sec:calibration}.
\end{itemize}
Many NCMs depend on segment parameters to be estimated, e.g. the $z$-score requires knowledge of the mean and variance.  To satisfy the properties of a good atypicity function, the estimators need to satisfy the following requirements:
\begin{enumerate}
	\item the estimator should be robust to anomalies in the training set: the estimation  should not be affected by the presence of anomalies in the training set.
	\item The estimator should be efficient \cite{koyuncu2009efficient}. High precision estimation of the parameter should be obtained with a minimal number of samples.
\end{enumerate}

For example, it is shown in the Supplementary Material, in Section 3, that when computing the $z$-score, the classical standard deviation estimators, MLE and MAD,  are underperforming. In this paper, the BW (Biweight midvariance) estimator \cite{shoemaker1982robust} is used to implement the scoring function.

\subsection{Calibration set}
\label{sec:calibration}
Section \ref{sec:high-level} introduces the notion of calibration set by giving a high level description of the Breakpoint Based Anomaly Detector. It is a collection of data points representing the reference behavior, inspired by Conformal Anomaly Detection\cite{LaxhammarConformalanomalydetection2014, ishimtsev2017conformal}. It is built using data from the current segment, or from another segment in the history with a similar distribution probability compared to the current segment. The cardinality of the calibration set follows two constraints:
\begin{itemize}
	\item it should be large enough to ensure that the $p$-values are estimated with sufficient precision to generate a low false positive and false negative rate.
	\item it should not be too large to maximize the homogeneity of the data and to limit computation time.
\end{itemize}

\begin{figure}[tb]
	\centering
	\includegraphics[width=0.7\linewidth]{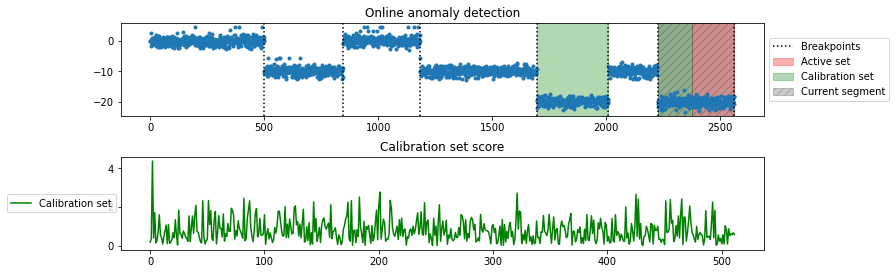}
	\caption{Illustration of the current segment, the active set and the calibration set.}
	\label{fig:calibration-set-illustration}
\end{figure}

As shown in Figure \ref{fig:calibration-set-illustration}, while data are collected online, the length of the current segment after the new breakpoint is too small to build the whole calibration set. By identifying similar segments and merging them to build the calibration set, the current segment can be completed with enough data points to estimate the $p$-values accurately.
Similar segments are found using a similarity function, like the Bhattacharyya distance proposed in \cite{bhattacharyya1943measure}. This similarity function is defined between two segments $S_1$, $S_2$ with means $\mu_1$ and $\mu_2$ and standard deviations $\sigma_1$ and $\sigma_2$ by:
\begin{equation}
	sim(\mathcal S_1, \mathcal S_2) = -\frac{1}{8\sigma^2}(\mu_1 - \mu_2)^2 - \frac{1}{2}\ln\frac{\sigma}{\sqrt{\sigma_1 \sigma_2}}
\end{equation}
The similarity function allows to sort all historical segments according to their similarity to the current segment. First, the similarity of each segment to the current segment is calculated. This allows to assign to each data point $X_u$ the variable that characterizes the similarity $sim_u$. By definition, the sequence $(sim_u)$ is constant on each segment $X_{\tau_i(t)}^{\tau_{i+1}(t)-1}$ and maximal on the current segment $X_{\hat b_t}^t$. 
\begin{align}
	\forall i \in \hat \llbracket 1, D_t\rrbracket, \forall u \in \llbracket \hat \tau_i(t), \hat \tau_{i+1}(t)-1\rrbracket, \quad sim_u = sim(X_{\hat \tau_i(t)}^{\hat \tau_{i+1}(t)-1}, X_{\hat b_t}^t) 
\end{align}
To build a calibration set of cardinality $n$, it is initialized using the scores of data points of the current segment that are not assigned to the active set. The data points scores with a ``normal'' status from the previous segments are added to the calibration set in descending order of similarity until $n$ scores are reached. 
After having described how a calibration set of a given cardinality $n$ is built, Appendix~\ref{sec:threshold} describes how the optimal cardinality $n$ is chosen.

\subsection{$p$-value estimation and threshold selection}
\label{sec:threshold}

As a reminder from Algorithm~\ref{alg:generalise-sbad}, the empirical $p$-values of each data point of the active set are computed using the calibration set. The threshold is chosen using the $p$-values of the active set to ensure the control of the FDR at a given level $\alpha$. Finally, the status of each data point of the active set is reevaluated comparing its $p$-value to the threshold.

In \cite{kronert2023fdr} we detail a new strategy for controlling the FDR of an anomaly detector in the online framework.
This goal is achieved by efficiently controlling the modified FDR criterion (mFDR) of subseries so that the FDR value of the full time series is controlled at the prescribed level $\alpha$.
To be more specific, \cite{kronert2023fdr} designs a modified version of the Benjamini-Hochberg procedure. Instead of applying BH to the active set with a slope $\alpha$, it is applied with a slope $\alpha'= \frac{\alpha}{1+\frac{1-\alpha}{m\pi}}$, where $m$ denotes the length of the active set, $\alpha$ is the desired global FDR level, and $\pi$ refers to the proportion of anomalies. Since $\alpha'$ depends on $\pi$, an estimation of $\pi$ (or expert knowledge) is required to detect anomalies. Some guidelines are provided in \cite{storey2004strong}. Notice that when $\pi$ is given, the fix threshold $\frac{\pi\alpha}{1+\pi - \alpha}$ control the FDR at level $\alpha$, this is equivalent to using BH with a subseries of length $m=1$.

\begin{figure}[tb]
	\centering
	\includegraphics[width=0.5\linewidth]{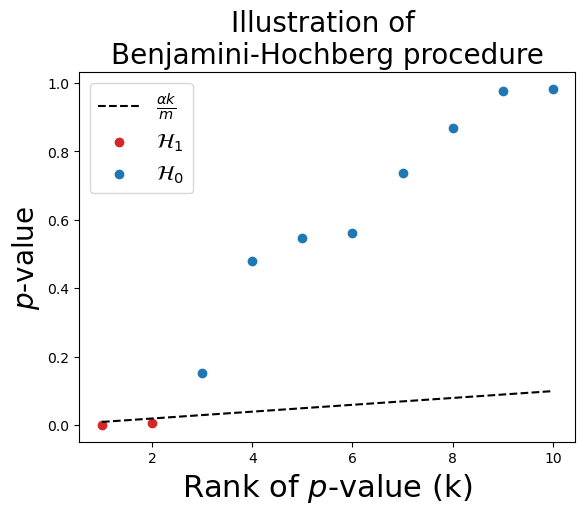}
	\caption{Example of Benjamini-Hochberg procedure.}
	\label{fig:illustrationbh}
\end{figure}

The calibration set is used to compute the $p$-values.
The FDR and the FNR of the modified BH procedure is very sensitive to the cardinality of the calibration set used to estimate the $p$-value. In \cite{kronert2023fdr}, we study under which conditions the cardinality of the calibration set ensures a control of the FDR. Given $m$ the cardinality of the active set and $\alpha'$ the modified slope for BH, the calibration set cardinality has to be chosen among:
\begin{equation}\label{eq:calsetcardi}
	n \in \lbrace \nu \frac{m}{\alpha'}  - 1, \quad \nu \in \mathbb{N}^*\rbrace
\end{equation}

As explained more deeply in \cite{kronert2023fdr}, the number of false negatives decreases with higher $\nu$. But a larger $\nu$ also increases the computation time, which can make any real-time decision difficult. We recommend to try different values of $\nu$, to monitor the decision time and to choose the largest $\nu$ which allows real time decisions.

\end{appendix}

\begin{supplement}
	\stitle{Supplementary Material for: Breakpoint based online anomaly detection}
	\sdescription{This Supplementary Material provides more details on the setup and results of the experiments presented in this paper.}
\end{supplement}
\end{document}